\documentclass[final,leqno,onefignum,onetabnum]{article}

\usepackage{doc}
\usepackage{color,soul}
\usepackage{latexsym}
\usepackage{amssymb}
\usepackage{amsmath}
\usepackage{amsthm}
\usepackage{bbm}
\usepackage{bm}
\usepackage{graphicx}
\usepackage{subfigure}
\usepackage{algorithm}
\usepackage{algorithmic}
\algsetup{linenosize=\small, linenodelimiter=.}

\newtheorem{theorem}{Theorem}

\newtheorem{corollary}[theorem]{Corollary}
\newtheorem{lemma}[theorem]{Lemma}
\newcommand{\R}{\mathbb R}

\newcommand{\E }{\mathbb E}
\renewcommand{\d}{{d}}

\begin{document}

\title{A mathematical framework for exact milestoning}
\author{David Aristoff\thanks{Department of Mathematics, Colorado State University, Fort Collins, CO}, Juan M. Bello-Rivas\thanks{Institute for Computational Engineering and Sciences, University of Texas at Austin, Austin, TX}, Ron Elber\thanks{Institute for Computational Engineering and Sciences, Department of Chemistry, University of Texas at Austin, Austin, TX}}
\date{}
\maketitle

\begin{abstract}
We give a mathematical framework for Exact Milestoning, a recently introduced algorithm for mapping a continuous time stochastic process into a Markov chain or semi-Markov process that can be efficiently simulated and analyzed. We generalize the setting of Exact Milestoning and give explicit error bounds for the error in the Milestoning equation for mean first passage times.
\end{abstract}




\section{Introduction}

Molecular Dynamics (MD) simulations, in which classical equations of motions are solved for molecular systems of significant complexity, have proven useful for interpreting and understanding many chemical and biological phenomena (for textbooks see~\cite{Schlick2010,Frenkel2002,Allen1989}). However, a significant limitation of MD is of time scales. Many molecular processes of interest occur on time scales significantly longer than the temporal scales accessible to straightforward simulations. For example, permeation of molecules through membranes can take hours~\cite{Cardenas2012} while MD is usually restricted to the microseconds time scale. One approach to extend simulation times is to use faster hardware~\cite{Shaw2008,Stone2007,Ruymgaart2011}. Other approaches focus on developing theories and algorithms for long time phenomena. Most of the emphasis has been on methodologies for activated processes with a single dominant barrier, as in Transition Path Sampling~\cite{Dellago2002,Bolhuis2002,Dellago1998}. Approaches for dynamics on rough energy landscapes, and for more general and/or diffusive dynamics, have also been developed~\cite{Metzner2009,Sarich2010,Chodera2007,Swenson2014,Moroni2004}. The techniques of Exact Milestoning~\cite{Bello-Rivas2015} and Milestoning~\cite{Faradjian2004} belong to the last category. They are theories and algorithms to accelerate trajectory calculations of kinetics and thermodynamics in complex molecular systems. The acceleration is based on the use of a large number of short trajectories instead of complete trajectories between reactants and products (Figure~\ref{fig:basic-figure}). The simulation of short trajectories is trivial to implement in parallel, making the formulation efficient to use on modern computing resources. Moreover, the use of short trajectories makes it possible to enhance sampling of improbable but important events by initiating the short trajectories near bottlenecks of reactions. A challenge is how to start the short trajectories, and how to analyze the result to obtain correct long time behavior.

While Milestoning is an approximate procedure, it shares the same philosophy and core algorithm as the Exact Milestoning approach. In both algorithms the phase space $\Omega$ is partitioned by hypersurfaces, which we call milestones 
$M \subset \Omega$, into cells. The short trajectories are initiated on milestones and are terminated the first time they reach a neighboring milestone (Figure~\ref{fig:basic-figure}). The short trajectories can 
be simulated in parallel. 

Milestoning uses an approximate distribution for the initial conditions of the trajectories at the hypersurfaces. The results are then analyzed within the Milestoning theory. The approximation is typically the (normalized) canonical distribution restricted to the 
milestone interface $M$. In Exact Milestoning the distribution of hitting points at the interface is estimated numerically by iteratively computing trajectory fragments between milestones. In a straightforward implementation of the iterations 
(see also~\cite{Bello-Rivas2015}) the final phase points of trajectories that were terminated on one milestone are continued until they hit another milestone. This type of trajectory continuation procedure is also used in Non-Equilibrium
Umbrella Sampling~\cite{Warmflash2007} and Trajectory Tilting~\cite{Vanden-Eijnden2009}. The continuation does not mean that full trajectories from reactants to products are computed. The calculations stop when the stationary distribution at the interface, or observables of interest, converge. In practice, and depending of course on the initial guess, the calculation ends significantly earlier than 
complete trajectories from $R$ to $P$ are 
computed. The fast convergence of 
the iterations leads to significant computational savings.

A number of other algorithms build on the use of short trajectories to estimate long time kinetics by ``patching'' these short trajectories at milestones or interfaces. These technologies include the Weighted Ensemble (WE)
~\cite{Zhang2010,Huber1996}, Transition Interface Sampling (TIS)~\cite{van-Erp}, Partial Path Transition Interface Sampling (PPTIS)~\cite{Moroni2004}, Forward Flux Sampling (FFS)~\cite{Allen2006}, Non-Equilibrium Umbrella Sampling (NEUS)~\cite{Warmflash2007}, Trajectory Tilting~\cite{Vanden-Eijnden2009}, and Boxed Molecular Dynamics (BMD)
~\cite{Glowacki2011}. Some of these techniques are similar; however, many subtle differences remain. Some of the differences are as follows. WE is the only method that makes it necessary to use stochastic dynamics. The trajectory sampling in NEUS, Trajectory Tilting and Exact Milestoning is similar, even though the theories are quite different. Exact Milestoning 
allows for the calculations of all the moments of the first passage time~\cite{Bello-Rivas2015}, a result which is not available for other technologies. Boxed Molecular Dynamics, Milestoning and PPTIS are approximate methods leading to greater efficiency. TIS, PPTIS, FFS, and BMD are focused on one-dimensional reaction 
coordinates. Other technologies (e.g. WE, Milestoning, NEUS, and Trajectory Tilting) focus on a space of one or several coarse variables.  

Hence, the overall scopes of these techniques differ significantly, which make direct comparison between them less obvious. We have compared in the past the accuracy and efficiency of the methods of Milestoning and Exact Milestoning with Forward Flux~\cite{Bello-Rivas2015,Cardenas2012_2}. Forward Flux is one of the closest algorithms (in one dimension) to Milestoning and Exact Milestoning. Numerous examples of kinetics of molecular systems studied with Milestoning were published
~\cite{Cardenas2012,Cardenas2013,Kirmizialtin2012,
Jas2012,Kreuzer2012,Elber2010,Elber2009,West2007,Elber2007}. We have also discussed extensively the features of alternate technologies that exploit trajectory fragments
~\cite{West2007,Majek2010}.

The Milestoning theory has not yet been subject to rigorous mathematical analysis, which is the goal of the present manuscript. In this manuscript we show that the Exact Milestoning method can be derived and analyzed in the framework of probability theory. The result is a useful link between physical intuition and a more formal approach. Readers that are interested in the efficiency of the algorithm on concrete examples, and comparison to other technologies, are referred to the sources mentioned the above paragraph.

\begin{figure}[ht]
  \centering
  \includegraphics[width=0.6\linewidth]{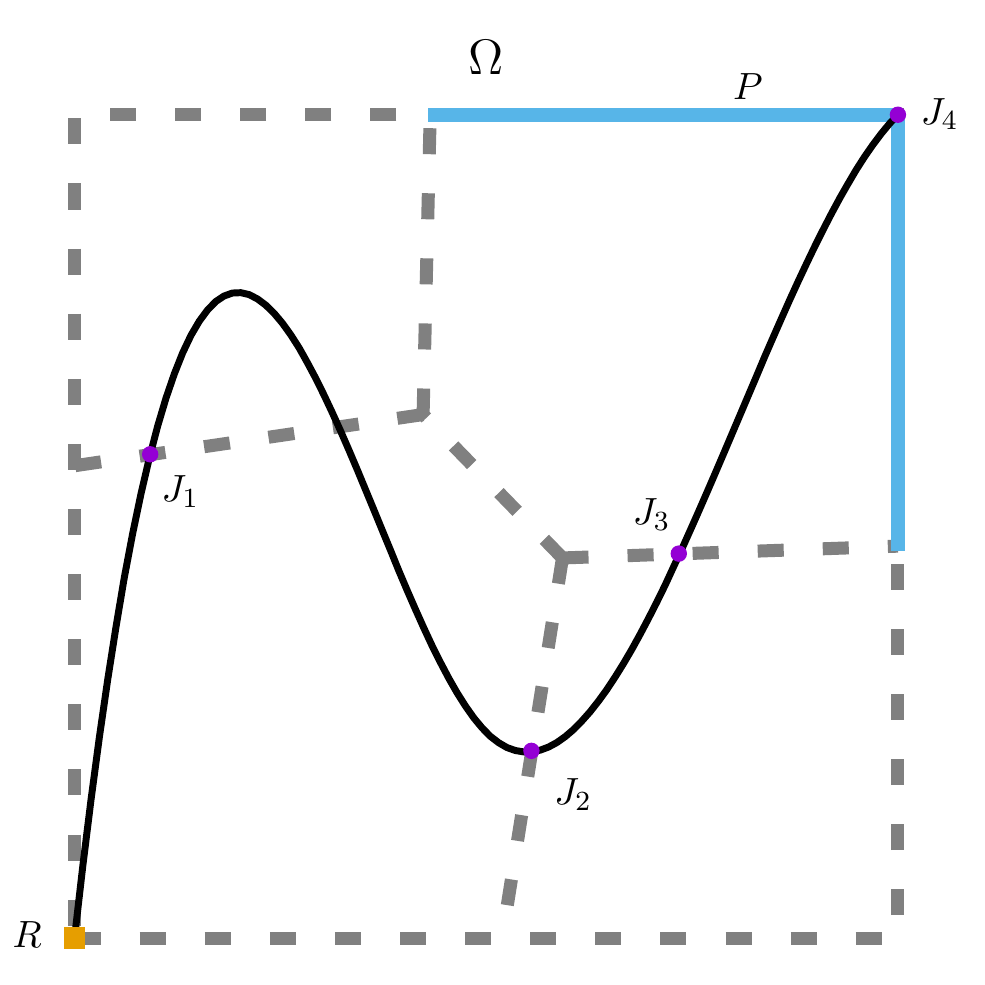}
  \caption{Representation of the state space $\Omega$ and
  the milestones. Each milestone is one 
  of the line segments
traced by dashed grey lines. The reactant state $R$ is highlighted as a square dot in the left-bottom corner while the product state $P$ is comprised of the two line segments shown in blue at the upper-right corner.  A particular realization of a long trajectory appears as a continuous black line and the corresponding values of $(J_n)$ are marked with dots.}
  \label{fig:basic-figure}
\end{figure}

This article is organized as follows. In
Section~\ref{sec:setup}, we describe the
setting for Exact Milestoning and
introduce notation used throughout.
In Section~\ref{sec:stationary},
we show existence of and convergence
to a stationary flux under very general
conditions.
In Section~\ref{sec:algorithm} we state
precisely the Exact Milestoning algorithm~\cite{Bello-Rivas2015}.
In Section~\ref{sec:error},
we establish conditions under which
convergence to the stationary flux is
consistent in the presence of numerical
error (Lemma~\ref{geom_ergodic} and
Theorem~\ref{feller}), and we give a natural upper
bound for the numerical error arising in
Exact Milestoning (Theorem~\ref{errors}).
Finally, in Section~\ref{sec:examples} we consider some
instructive examples.

\section{Setup and notation}\label{sec:setup}

\subsection{The dynamics and MFPT}

In Milestoning we spatially
coarse-grain a dynamics $(X_t)$. 
{The basic 
idea is to stop and start 
trajectories on certain interfaces, 
called milestones, and then 
reconstruct functions of $(X_t)$ using these short 
trajectories, which 
can be efficiently simulated 
in parallel. We assume 
here the dynamics is 
stochastic, and focus 
on using Milestoning for the
efficient computation of}
{\em mean first passage times} 
{(MFPTs)
of $(X_t)$, although similar ideas 
can be used to compute other 
non-equilibrium quantities.}

{To make our arguments we need 
some assumptions on 
$(X_t)$.
We let $(X_t)$ be a 
time homogeneous strong
Markov process with c\`adl\`ag paths
taking values in
a standard Borel space $\Omega$. 
These assumptions allow us to
stop and restart $(X_t)$ on the 
milestones 
without knowing its history.
In applications, usually 
$(X_t)$ is Langevin or overdamped 
Langevin dynamics, and $\Omega$ 
is a subset of Euclidean space.}

We write
${\mathbb P}$, ${\mathbb E}$
for all probability measures
and expectations, with superscripts 
${\mathbb P}^x$ (resp. ${\mathbb P}^\xi$) 
to indicate a starting point $x$ 
(resp. distribution $\xi$). 
The symbol $\sim$ will indicate equality in 
probability law.
We will use 
the words {\em distribution} and 
{\em probability measure} interchangeably.
Total variation norm will be 
denoted by $\|\cdot\|_{TV}$. 
Our analysis below will mostly 
take place in an idealized 
setting where we assume infinite 
sampling on the 
milestones. In this setting, 
distributions are smooth (if 
state space is continuous) and 
the total variation norm is 
the appropriate one. 

Recall we are interested in computing
the MFPT of $(X_t)$
from a reactant set $R$
to a product set $P$. Throughout we
consider fixed disjoint product and
reactant sets $P,R \subset \Omega$.
When $R$ is not a single point, we
will start $(X_t)$ from a fixed 
probability measure $\rho$ on $R$. 
{If $R$ is a single point, 
$\rho = \delta_R$, the delta 
distribution at $R$.} As 
discussed above, 
Milestoning allows for an efficient
computation of the MFPT of $(X_t)$
to $P$, starting at $\rho$. It 
is useful to think of $P$ as a sink,
and $R$ as a source for $(X_t)$.
More precisely, {\em we assume that when
$(X_t)$ reaches $P$, it immediately
restarts on $R$ according to $\rho$.}
Obviously, this 
assumption has no effect on the MFPT
to $P$.  It will be useful, 
however, for
computational and
theoretical considerations.

Many of the results below 
follow from well-known theorems in probability 
theory. However, because of 
the special source-sink structure of $(X_t)$, 
simpler proofs are often available, 
and we include them for clarity 
and completeness.

\subsection{The milestones and semi-Markov
viewpoint}
We write $M \subset \Omega$ for 
the space of milestones used 
for parallelizing the computation 
of the MFPT. 
Each point $x \in M$ belongs to a
{\em milestone} $M_x \subset M$. Thus, 
$M$ is the union of all the milestones. 
We assume there are 
finitely many milestones, 
each of which is a closed
set. Moreover, 
we demand that $(X_t)$ passes through 
the intersection of two milestones 
with probability $0$ -- thus, 
$(X_t)$ can only cross one 
milestone at a time. This 
can be accomplished for Langevin 
or overdamped Langevin dynamics
by taking the 
milestones to be codimension 1 
with pairwise intersections of 
codimension 2 or larger; see Figure~\ref{fig:basic-figure}.
The sets $P$ and $R$ will be two
of the milestones. We always 
start $(X_t)$ on $M$.

By following the sequence 
of milestones crossed by $(X_t)$, 
we obtain a sequence of points $(J_n)$ 
in $M$. See Figure~\ref{fig:basic-figure}. 
We now describe $(J_n)$ more precisely. 
Let $\theta_n$ be the $n$th milestone 
crossing time for $(X_t)$, defined 
recursively by $\theta_0 = 0$ and 
\begin{equation*}
\hbox{if }X_{\theta_n} = x,
\hbox{ then } \theta_{n+1} :=
\inf\{t > \theta_n\,:\, X_t \in M_y\text{ for some } M_y\ne M_x\}.
\end{equation*}
Note that by a milestone 
crossing, we mean a crossing of 
a milestone different from the 
previous one. The sequence of milestone crossing points is $J_n = X_{\theta_n}$.

We show now that $(X_t)$ can be 
partially reconstructed 
from $(J_n)$ and $(\theta_n)$. 
Let $(Y_t)$ be 
defined\footnote{When $(Y_t)$ has a probability density, it corresponds to the
density $p(x,t)$ 
from~\cite{Bello-Rivas2015} 
for the
last milestone point passed.} by setting $Y_t = J_n$ 
whenever $\theta_n \le t < \theta_{n+1}$. 
Then $(X_t)$ and $(Y_t)$ 
agree at each milestone 
crossing time $t=\theta_n$, $n=0,1,2,\ldots$ 
and 
$(Y_t)$ is obtained from $(X_t)$ by throwing 
away the path of $(X_t)$ between 
milestone crossings, keeping only the endpoints. 
It follows that $(X_t)$ and $(Y_t)$ have the same
MFPT to $P$. Thus, for our purposes 
it is enough to study $(Y_t)$. We 
note that $(Y_t)$, like $(X_t)$,  
immediately restarts at $\rho$ upon 
reaching $P$.

By our assumptions above, $(J_n)$ 
is a Markov chain on $M$, and $(Y_t)$ 
is a semi-Markov process on $M$, meaning 
it has the Markov property at 
crossing times. We 
write $K(x,dy)$ for the transition 
kernel of $(J_n)$. Thus, if 
the initial distribution of $(J_n)$ 
is $J_0 \sim \xi$, 
then the distribution at time 
$n$ is ${\mathbb P}^\xi(J_n \in \cdot) 
= \xi K^n$. We also write 
$\xi K^n f := {\mathbb E}^\xi[f(J_n)]$ 
and $\xi f := \int_M f(x) \,\xi(dx)$
for suitable functions $f$.

The following
notation will be needed.
For $x \in M$, define local  
first passage times
\begin{equation*}
\tau_M^x = \inf\{t > 0\,:\, Y_t \in M_y \text{ for some }M_y \ne M_x\}.
\end{equation*}
Thus, $\tau_M^x$ is the first time for $(Y_t)$ to
cross some milestone other than $M_x$,
starting at $Y_0 = x$. In particular,  
if $X_{\theta_{n-1}} = x$, 
then $\theta_n \sim \theta_{n-1} + \tau_M^x$.
We also define $\tau_P$ to be 
the first time to cross $P$ 
and $\sigma_P$ the number of 
crossings before reaching $P$,
\begin{equation*}
\tau_P = \inf\{t > 0\,:\, Y_t \in P\}, 
\qquad \sigma_P =  \min\{n \ge 0\,:\,
J_{n} \in P\}.
\end{equation*}
We 
are interested in
${\mathbb E}^\rho[\tau_P]$, 
the MFPT from $\rho$ to $P$.

\section{Invariant measure and MFPT}\label{sec:stationary}

\subsection{Stationary distribution on the milestones}

The MFPT 
will be estimated via short 
trajectories between milestones. 
An important ingredient is the 
correct starting 
distribution for these trajectories. 
Exact Milestoning 
makes use of a {\em stationary flux} of
$(X_t)$ on the milestones, 
which corresponds\footnote{Our $\mu$ is the same as the
appropriately normalized
stationary flux $q$ in other
Milestoning papers. 
We use $\mu$ 
instead of $q$ to emphasize that here it 
is a probability measure, not a density.} to the 
stationary distribution $\mu$ of 
$(J_n)$. It is worth noting that 
Milestoning can also be made
exact by choosing 
milestones as 
isocommittor 
surfaces~\cite{Vanden-Eijnden2008}. 
The advantage of the formulation 
here is that the milestones can 
be arbitrary.

Some assumption is required to guarantee
the existence of a stationary flux.
We adopt the following 
sufficient condition, which we assume
holds throughout:
\begin{equation*}
{\mathbb E}^\xi [\tau_P] \text{ and }
{\mathbb E}^\xi [\sigma_P] \text{ are finite
for all probability measures }\xi 
\text{ on }M.
\end{equation*}
This ensures that $(Y_t)$ reaches 
$P$ in finite expected time and 
does not have infinitely many 
milestone crossings in finite time. 
The condition
can be readily verified in the
standard settings for milestoning
discussed above. Using this assumption 
and the source-sink structure of the 
dynamics -- namely, that $(Y_t)$ immediately 
restarts at $\rho$ upon reaching $P$ -- 
we show in Theorem~\ref{mu} 
below that $\mu$ exists. 
\begin{theorem}\label{mu}
$(J_n)$ has an
invariant probability
measure $\mu$ defined by
\begin{equation*}
\mu(\cdot) := {\mathbb E}^{\rho}\left[\sum_{n=0}^{\sigma_P} {\mathbbm{1}}_{\{J_n \in\, \cdot\}}\right] {\mathbb E}^\rho[\sigma_P+1]^{-1}.
\end{equation*}
where ${\mathbbm{1}}_{\{J_n \in C\}} = 1$ 
if $J_n \in C$ and otherwise ${\mathbbm{1}}_{\{J_n \in C\}}=0$.
\end{theorem}
\begin{proof}
Define $\nu(\cdot) = {\mathbb E}^{\rho}\left[\sum_{n=0}^{\sigma_P} {\mathbbm{1}}_{\{J_n \in\, \cdot\}}\right]$ and observe that
\begin{align*}
\nu(\cdot)
= \sum_{n=0}^\infty \sum_{m=0}^{n}
{\mathbb P}^{\rho}(J_m \in \cdot\,|\, \sigma_P = n){\mathbb P}^\rho(\sigma_P = n)= \sum_{n=0}^\infty {\mathbb P}^\rho(J_n \in \cdot,\, \sigma_P \ge n).
\end{align*}
If $C \cap R = \emptyset$, by bounded convergence,
\begin{align*}
\int_M \nu(dx)K(x,C) =
\sum_{n=0}^\infty
{\mathbb P}^\rho(J_{n+1} \in C,\, \sigma_P \ge n)
= \sum_{n=0}^\infty
{\mathbb P}^\rho(J_{n} \in C,\, \sigma_P \ge n) = \nu(C),
\end{align*}
where the second equality
uses ${\mathbb P}^\rho(J_0 \in C) = 0$ and
 $J_{n+1} \notin R
\Rightarrow \sigma_P \ne n$.
If $C \subset R$,
\begin{align*}
\int_M \nu(dx)K(x,C)
&= \sum_{n=0}^\infty
{\mathbb P}^\rho(J_{n+1} \in C,\, \sigma_P = n) + \sum_{n=0}^\infty{\mathbb P}^\rho(J_{n+1} \in C,\, \sigma_P \ge n+1) \\
&= \rho(C) - {\mathbb P}^\rho(J_0 \in C, \sigma_P \ge 0) + \sum_{n=0}^\infty{\mathbb P}^\rho(J_{n} \in C,\, \sigma_P \ge n) = \nu(C).
\end{align*}
\end{proof}

We will show below that $(J_n)$ 
converges to $\mu$ under appropriate 
conditions. In that case $\mu$ 
is unique and we will call
$\mu$ the {\em stationary distribution} 
of $(J_n)$.
A successful application of 
Exact Milestoning 
will require some technique 
for sampling $\mu$. 
The algorithm we present (Algorithm 1 
below) is based on convergence of 
the distribution of $(J_n)$ to $\mu$ 
in total variation. We 
demonstrate this convergence in 
Theorem~\ref{conv} under an additional 
assumption on $(J_n)$. 

It is worth noting that 
the proof of Theorem~\ref{mu} leads to
the following representation of $\mu$
as a Neumann series. The 
representation is given in Corollary~\ref{neumann} below. 
This representation
can be used, in principle, to sample $\mu$ 
without additional assumptions on $(J_n)$.
The Neumann series is written 
in terms of the 
{\em transient kernel}
\begin{equation}\label{trans}
{\bar K}(x,dy) = \begin{cases} K(x,dy), & x \notin P\\
0, & x \in P\end{cases}.
\end{equation}
${\bar K}(x,dy)$  
corresponds to a modified version of $(J_n)$ that 
is absorbed (killed) on $P$.

\begin{corollary}\label{neumann}
We have
\begin{equation}\label{neumann_c}
\lim_{n \to \infty} \|\nu(M)^{-1}\sum_{i=0}^{n-1} \rho {\bar K}^i - \mu\|_{TV} = 0.
\end{equation}
\end{corollary}
\begin{proof}
Recall that $\mu = \nu/\nu(M)$ where
\begin{equation*}
\nu(\cdot) = \sum_{n=0}^\infty
{\mathbb P}^\rho(J_n \in \cdot,\, \sigma_P \ge n) = \sum_{n=0}^\infty \rho {\bar K}^n.
\end{equation*}
Moreover,
\begin{equation}\label{neumann_error}
\sup_{|f| \le 1} \left|\nu(M)^{-1}\sum_{i=0}^{n-1} \
\rho {\bar K}^i f - \mu f\right|
 \le \nu(M)^{-1}\sum_{i=n}^\infty{\mathbb P}^\rho(\sigma_P \ge i),
\end{equation}
and the right hand side of~\eqref{neumann_error} is summable since by assumption
${\mathbb E}^\rho[\sigma_P] < \infty$.
\end{proof}

\subsection{Milestoning equation for
the MFPT}\label{milestoning_eq}

Equipped with an invariant
measure $\mu$, we are now able to state the Milestoning equation~\eqref{main_eq} for the MFPT. In Exact Milestoning, this equation
is used to efficiently compute
the MFPT. The algorithm is 
based on two principles: first, 
many trajectories can be simulated 
in parallel to estimate $\tau_M^x$ 
for various $x$; and second, 
the stationary distribution $\mu$ 
can be efficiently estimated through 
a technique based on power 
iteration. See the right hand 
side of equation~\eqref{main_eq} below. 

The gain in efficiency comes 
from the fact that the trajectories 
used to estimate $\tau_M^x$ are 
much shorter than trajectories 
from $R$ to $P$. Whether we can 
efficiently sample $\mu$ may
depend somewhat on whether we have 
a good initial guess. 
When $(X_t)$ is Langevin dynamics, 
we have found in some cases 
the canonical Gibbs 
distribution is a sufficiently 
good guess.
See~\cite{Bello-Rivas2015} and~\cite{Bello-Rivas2015b} for details and discussion.

\begin{theorem}\label{milestoning}
Let $\mu$ be defined as above. Then $\mu(P)> 0$ and
\begin{equation}\label{main_eq}
\mu(P){\mathbb E}^\rho[\tau_P] = \int_M \mu(dx){\mathbb E}^x[\tau_M^x] := {\mathbb E}^\mu[\tau_M].
\end{equation}
\end{theorem}
\begin{proof}
The assumption
${\mathbb E}^\rho[\sigma_P] < \infty$
shows that $\mu(P) > 0$.
For any $x \in M$,
\begin{align*}
{\mathbb E}^x[\tau_P]
&= \int_M {\mathbb E}^x\left[\tau_P\,|\,
Y_{\tau_M^x} = y\right]K(x,dy) \\
&= \int_M {\mathbb E}^x\left[\tau_M^x\,|\,
Y_{\tau_M^x} = y\right]K(x,dy)
+ \int_{M\setminus P} {\mathbb E}^x\left[\tau_P - \tau_M^x\,|\,
Y_{\tau_M^x} = y\right]K(x,dy)\\
&= {\mathbb E}^x[\tau_M^x]
+ \int_{M\setminus P} {\mathbb E}^y[\tau_P]
K(x,dy).
\end{align*}
Thus,
\begin{equation*}
{\mathbb E}^{\mu}[\tau_P]
= {\mathbb E}^{\mu}[\tau_M]
+ \int_{M\setminus P} \int_M
\mu(dx){\mathbb E}^y[\tau_P]K(x,dy)= {\mathbb E}^\mu[\tau_M]
+ \int_{M\setminus P} \mu(dy){\mathbb E}^y[\tau_P],
\end{equation*}
and so
\begin{equation*}
{\mathbb E}^\mu[\tau_M] = \int_P \mu(dy){\mathbb E}^y[\tau_P] = \mu(P){\mathbb E}^\rho[\tau_P].
\end{equation*}
\end{proof}

In Section~\ref{sec:algorithm}
below we present the Exact Milestoning 
algorithm (Algorithm 1) 
recently used in~\cite{Bello-Rivas2015} 
and~\cite{Bello-Rivas2015b}. The algorithm 
uses a technique which combines 
coarse-graining and power iteration 
to sample $\mu$. Consistency 
of power iteration algorithms 
are justified via 
Theorem~\ref{conv} below, where 
we show $\xi K^n \to \mu$ as 
$n\to \infty$. Though we emphasize
that there are a range of possibilities
for sampling~$\mu$ (for 
example, algorithms based 
on~\eqref{neumann_c} or~\eqref{erg} 
below) 
we note that Algorithm 1 
was shown
to be efficient for
computing the MFPT in the
entropic barrier example
of~\cite{Bello-Rivas2015} and
the random energy landscapes
example of~\cite{Bello-Rivas2015b}.

\subsection{Convergence to stationarity}
\label{sec:conv}

In this section we justify 
the consistency of 
power iteration-based 
methods for sampling $\mu$ 
by showing that $\xi K^n$ 
converges to $\mu$ in total 
variation norm as $n\to \infty$. 
The theorem requires an extra assumption -- 
aperiodicity of the jump chain $(J_n)$.

\begin{theorem}\label{conv}
Suppose that $(J_n)$ is aperiodic
in the following sense:
\begin{equation}\label{nonlattice}
\hbox{g.c.d.}\, \{n \ge 1\,:\,{\mathbb P}^\rho(\sigma_P = n-1) > 0\} = 1.
\end{equation}
Then for all probability measures $\xi$
on $M$,
\begin{equation}\label{ergodic}
\lim_{n\to \infty} \|{\mathbb P}^\xi(J_n \in \cdot) - \mu\|_{TV} \equiv \lim_{n\to \infty} \|\xi K^n -\mu\|_{TV} = 0.
\end{equation}
In particular, $\mu$ is unique.
\end{theorem}
\begin{proof}
We use a simple coupling argument. Let $(H_n)$ be an independent copy of
$(J_n)$ and let $J_0 \sim \xi$ and
$H_0 \sim \mu$.  For $n \ge 0$, let $S_n$ (resp. $T_n$) be the times at which $(J_n)$ (resp. $(H_n)$) hit $P$ for the $(n+1)$st time. Then $S_{n+1}-S_{n}$, ${n \ge 0}$, 
are iid random variables with finite 
expected value and nonlattice 
distribution, and 
$(S_{n+1}-S_n)_{n\ge 0} \sim 
(T_{n+1}-T_n)_{n\ge 0}$. 
It follows that $(S_n-T_n)_{n\ge 0}$ is a mean zero random
walk with nonlattice step distribution.
Thus, its first time to hit $0$ is finite
almost surely. So
\begin{equation*}
\zeta := \inf\{n \ge 0\,:\, J_n \in P,\, H_n \in P\}
\end{equation*}
obeys
${\mathbb P}(\zeta \ge n) \to 0$ as $n \to \infty$. Note that
$J_n \sim H_n$
whenever $\zeta < n$. Thus
\begin{equation*}
|{\mathbb P}^\xi(J_n \in C) - {\mathbb P}^\mu(H_n \in C)| \le 2{\mathbb P}(\zeta \ge n).
\end{equation*}
Since $\mu$ is stationary for $(H_n)$ we have ${\mathbb P}^\mu(H_n \in C) = \mu(C)$.
Now
\begin{equation*}
\|{\mathbb P}^{\xi}(J_n \in \cdot) - \mu\|_{TV} = \sup_{C \subset M} |{\mathbb P}^\xi(J_n \in C) - \mu(C)| \le 2{\mathbb P}(\zeta \ge n),
\end{equation*}
which establishes the convergence result.
To see uniqueness, suppose $\xi$ is another invariant probability measure for $(J_n)$;
then the last display becomes $\|\xi-\mu\|_{TV} \le 2{\mathbb P}(\zeta \ge n)$. Letting
$n \to \infty$ shows that $\xi \sim \mu$.
\end{proof}

We now consider a class of 
problems where there 
is a smooth one-dimensional {\em reaction 
coordinate} $\psi:\Omega \to [0,1]$ 
tracking progress of $(X_t)$ from 
$R$ to $P$. In this case $\psi|_R \equiv 0$, 
$\psi|_P \equiv 1$, the milestones 
$M_1,\ldots,M_m$ are disjoint level sets of $\psi$, 
and $R = M_1$, $P = M_m$. The jump 
chain $(J_n)$ can only hop between 
neighboring milestones, unless it is 
at $P$. That is, if $J_n \in M_i$ for 
$i \notin \{1,m\}$, then $J_{n+1} \in M_{i-1}$ or $J_n \in M_{i+1}$; 
if $J_n \in M_1$ then $J_{n+1} \in M_2$; 
and if $J_n \in M_m$, then $J_{n+1} \in M_1$. 
Suppose that if $J_n \in M_i$ for 
$i \notin \{1,m\}$, then $J_{n+1} \in M_{i-1}$ 
with probability in $(0,1)$. Then the aperiodicity 
assumption~\eqref{nonlattice} is 
satisfied if and 
only if $m$ is odd. This is due to 
the fact that, if $J_0 \in M_1$, then 
$J_{m-1} \in M_m$ and $J_{m+1} \in M_m$ 
with positive probability, and $m$ and $m+2$ 
are coprime when $m$ is odd. 
On the other hand, if $m$ is even then 
the conclusion of Theorem~\ref{conv} 
cannot hold. To see this, let $m$ 
be even and suppose $J_0$ is supported 
in an odd-indexed milestone. 
Then $J_{2n}$ is always supported on an 
odd-indexed milestone, 
while $J_{2n+1}$ is always supported on an 
even-indexed milestone.

Theorem~\ref{conv} estabishes 
convergence the distribution of $J_n$ 
to $\mu$ in total variation norm. 
Even when $(J_n)$ is not aperiodic, 
it converges in 
a time-averaged sense. 
Thus, problems in sampling $\mu$ 
arising from aperiodicity 
can be managed by averaging over 
time. More precisely, we have 
the following version 
of the Birkhoff ergodic theorem:

\begin{theorem}\label{cesaro}
Let $J_0 \sim \xi$, with
$\xi$ a probability measure on $M$. 
For bounded measurable $f:M \to {\mathbb R}$,
\begin{equation}\label{erg}
\lim_{n\to \infty} \frac{1}{n}\sum_{i=0}^{n-1} f(J_i)
\stackrel{a.s.}{=} \int_M f\,d\mu \equiv \mu f.
\end{equation}
\end{theorem}
\begin{proof}
Let $S_n$ be the times at which 
$(J_n)$ hits $P$ for the $(n+1)$st 
time, and define 
\begin{equation*}
f_n = \sum_{i=S_n+1}^{S_{n+1}} f(J_i).
\end{equation*}
Note that $f_n$, $n\ge 0$, are 
iid.
Let $k(n) = \max\{k\,:\,S_k \le n\}$
and write 
\begin{equation*}
\frac{1}{n}\sum_{i=0}^{n-1} f(J_i)
= \frac{1}{n}\sum_{i=0}^{S_0} f(J_i)
+ \frac{1}{n}\sum_{i=0}^{k(n)-1} f_i
+ \frac{1}{n}\sum_{i=S_{k(n)}+1}^n f(J_i)
\end{equation*}
Since $(J_n)$ hits $P$ in finite time 
a.s., $n-S_{k(n)}$ and $S_0$ 
are finite a.s. Thus, 
\begin{equation*}
\lim_{n\to \infty}\frac{1}{n}\sum_{i=0}^{n-1} f(J_i) 
= \lim_{n\to \infty} 
\frac{k(n)-1}{n}\frac{1}{k(n)-1}\sum_{i=0}^{k(n)-1} f_i.
\end{equation*}
Notice $R_n:= S_{n+1}-S_n$, $n\ge 0$ are iid 
with finite expectation and
\begin{equation*}
\frac{R_0+\ldots+R_{k(n)-1}}{k(n)} \le \frac{n-S_0}{k(n)} 
\le \frac{R_0+\ldots+R_{k(n)}}{k(n)}.
\end{equation*}
By the previous two displays and the law of 
large numbers,
\begin{equation*}
\lim_{n\to \infty}\frac{1}{n} \sum_{i=0}^{n-1} f(J_i) 
= \frac{{\mathbb E}[f_0]}{{\mathbb E}[R_0]} 
= {\mathbb E}^\rho[\sigma_P+1]^{-1}
\sum_{i=0}^{\sigma_P}{\mathbb E}^\rho[f(J_i)] =\frac{\nu f}{\nu(M)} = \mu f,
\end{equation*} 
with $\nu$ defined as in Theorem~\ref{conv}.
\end{proof}

Markov chains for which the conclusion 
of Theorem~\ref{conv} hold are 
called {\em Harris ergodic}. 
It is worth noting that a slightly 
stronger aperiodicity condition
leads to a limit for the 
distribution of $Y_t$. 
More precisely, 
suppose~\eqref{nonlattice} holds and 
for each $x \in M\setminus P$ and 
$y \in M$, ${\mathbb P}^x(\tau_1 \in \cdot\,|\,J_1 = y)$ is nonlattice. Then for any 
$C \subset M$ and 
$\mu$-a.e. $x$, 
\begin{equation}\label{stat_prob}
\lim_{t\to \infty} {\mathbb P}^x(Y_t \in C) 
= \frac{\int_C \mu(dy){\mathbb E}^y[\tau_M^y]}{
\int_M \mu(dy){\mathbb E}^y[\tau_M^y]}.
\end{equation}
See~\cite{Alsmeyer1994} for details\footnote{When the right hand side of~\eqref{stat_prob} 
has a density, it 
is the same as 
the stationary probability density $p(x)$ 
in~\cite{Bello-Rivas2015} for 
the last milestone point 
passed.} and proof.

\section{Exact Milestoning algorithm}\label{sec:algorithm}

We now describe in detail
an algorithm for sampling $\mu$ 
and the MFPT 
${\mathbb E}^\rho[\tau_P]$, 
used successfully 
in~\cite{Bello-Rivas2015} and~\cite{Bello-Rivas2015b}. We assume 
throughout this section that the conclusion of 
Theorem~\ref{conv} holds.
Let $\xi$ be an initial guess for $\mu$.
(If $(X_t)$ is Brownian
or Langevin dynamics, we usually
take $\xi$ to be the canonical
Gibbs distribution.) 
We write $M_i$ for the distinct milestones,
so that $M = \cup_i M_i$.
The algorithm
will produce approximations
\begin{equation*}
\xi \equiv \mu^{(0)}, \mu^{(1)}, \mu^{(2)},\ldots 
\end{equation*}
of $\mu$. Let $\mu_i^{(n)}$
be the {\em non-normalized} 
restriction of $\mu^{(n)}$ to
$M_i$, and define
\begin{equation*}
{\mathbb E}^{\mu_i^{(n)}}[\tau_M] :=
\mu^{(n)}(M_i)^{-1}\int_{M_i} \mu_i^{(n)}(dx){\mathbb E}^x[\tau_M^x],
\end{equation*}
For $C \subset M_j$ we will also use
the notation
\begin{equation*}
  a_{ij}^{(n)}(C) = \mu^{(n-1)}(M_i)^{-1} \int_{M_i} \mu_i^{(n-1)}(dx) \, K(x,C).
\end{equation*}
Below we think of $a_{ij}^{(n)}$
and $\mu_i^{(n)}$ as either
distributions or densities. The
$a_{ij}^{(n)}$ are obtained from trajectory fragments 
between milestone crossings.
A simple Monte Carlo scheme for 
estimating these distributions is as follows. Let $x_1,\ldots,x_L$ be 
iid samples from the distribution 
$\mu_i^{(n-1)}/\mu^{(n-1)}(M_i)$. 
Starting at each $x_\ell \in M_i$, simulate $(X_t)$ until it crosses the next milestone, say at the point $y_\ell \in M_j$. 
If we idealize by assuming the simulation 
of $(X_t)$ is 
done exactly, then by Chebyshev's inequality,
\begin{equation*}
{\mathbb P}\left(\left|a_{ij}^{(n)}(C) - \int_C \frac{1}{L}\sum_{\ell=1}^L 
\delta_{y_\ell}(dy)\right| > \epsilon\right) \le \frac{a_{ij}^{(n)}(C)-a_{ij}^{(n)}(C)^2}{L\epsilon^2}, 
\end{equation*}
where $\delta_y$ is the Dirac delta distribution 
at $y$. We therefore write, for $y \in M$, 
\begin{equation}\label{aij}
  a^{(n)}_{ij}(y) \approx \frac{1}{L} \sum_{\ell = 1}^{L} {\tilde \delta}_{y_\ell}(y),
\end{equation}
where ${\tilde \delta}_{y_\ell}$ is either some suitable {approximation to the identity} at $y_\ell$, or simply a delta function at $y_\ell$. Thus, in Algorithm~1 we think of 
$a_{ij}^{(n)}$ and $\mu_i^{(n-1)}$ as either 
densities in the 
former case, or as distributions in the latter. 
The local mean first passage times (i.e., the times between successive
milestone crossings) are approximated by the sample means
\begin{equation*}
  \E^{\mu_i^{(n-1)}}[ \tau_M ] \approx \frac{1}{L} \sum_{\ell = 1}^{L} \tau_M^{x_\ell}.
\end{equation*}
It is important to realize that we do not need to store the full coordinates of each $y_\ell$ in memory.
Instead, it suffices to use a data-structure that keeps track of the pairs $(y_\ell, M_j)$.
The actual coordinates of each point can be written to disk and read from it as needed.
\begin{algorithm}[ht]
  \caption{Exact Milestoning algorithm.}
  \label{alg:exact-milestoning}
  \begin{algorithmic}[]
    \REQUIRE Milestones $M = \cup_{j = 1}^m M_j$, initial guess $\xi$, and tolerance $\varepsilon > 0$ for the absolute error in the MFPT.
    \ENSURE Estimates for $\mu$, local MFPTs $\E^\mu[\tau_M]$, and overall MFPT $\E^\rho[ \tau_P ]$.
    \STATE $\mu^{(0)}$ $\leftarrow$ $\xi$
    \STATE $T^{(0)}$ $\leftarrow$ $+\infty$
    \FORALL{$n = 1, 2, \dotsc$}
    \FOR{$i = 1$ \TO $m$}
    \STATE Estimate $a^{(n)}_{ij}$ and $\E^{\mu_i^{(n-1)}}[\tau_M]$
    \STATE $\mathtt{A}^{(n)}_{ij}$ $\leftarrow$ $a^{(n)}_{ij} (M_j)$
    \ENDFOR
    \STATE Solve $\mathbf{w}^{\mathtt{T}} \mathtt{A} = \mathbf{w}^{\mathtt{T}}$ (with $\mathtt{A} = ( \mathtt{A}^{(n)}_{ij} ) \in \R^{m \times m}_{\ge 0}$ and $\mathbf{w} = (w_1, \dotsc, w_m) \in \R^{m}_{\ge 0}$)
    \FOR{$j = 1$ \TO $m$}
    \STATE $\mu^{(n)}_j$ $\leftarrow$ $\sum_{i = 1}^m w_i \, a^{(n)}_{ij}$
    \ENDFOR
    \STATE Normalize $\mu^{(n)}$
    \STATE $T^{(n)}$ $\leftarrow$ ${\mu(P)}^{-1} \, \E^{\mu^{(n-1)}}[\tau_M]$
      \IF{$| T^{(n)} - T^{(n-1)} | < \varepsilon$}
    \STATE \textbf{break}
    \ENDIF
    \ENDFOR
    \RETURN $(\mu^{(n)}, \E^{\mu^{(n-1)}}[\tau_M], T^{(n)})$
  \end{algorithmic}
\end{algorithm}

The eigenvalue problem in Algorithm~\ref{alg:exact-milestoning} involves a stochastic matrix $\mathtt{A} \in \mathbb{R}^{m \times m}_{\ge 0}$ that is sparse.
Indeed, the $i$-th row corresponds to milestone $M_i$ and may have only 
as many non-zero entries as the number of neighboring milestones $M_j$.
In practice, to solve the eigenvalue problem we can use efficient and accurate Krylov subspace solvers~\cite{Golub2013} such as Arnoldi iteration~\cite{Lehoucq1998} to obtain $\mathbf{w}$ without computing all the other eigenvectors.

{In Algorithm~\ref{alg:exact-milestoning}, if $\mathbf{w}_i := \mu_i^{(n-1)}(M_i)$ 
is used instead of the solution 
$\mathbf{w}$ to $\mathbf{w}^{\mathtt{T}} \mathtt{A} = \mathbf{w}^{\mathtt{T}}$, then the 
algorithm approximates $\mu$ by 
simple power iteration, $\mu^{(n)} = \xi K^n$. The reason 
for defining the weights as the solution 
to $\mathbf{w}^{\mathtt{T}} \mathtt{A} = \mathbf{w}^{\mathtt{T}}$ is practical: we have found 
that it gives faster convergence of the 
iterations, at no apparent cost to accuracy. 
It can be seen as a version of power 
iteration that uses coarse-graining.} 
See~\cite{Bello-Rivas2015,Bello-Rivas2015b} 
for applications of the algorithm in Exact 
Milestoning and~\cite{Golub2013,Lehoucq1998} for related 
discussions.

Finally, we mention the fact that pseudo-random number generators (PRNGs) can only produce a finite amount of pseudo-random numbers.
Once the maximum amount is reached, the generators may silently reuse the previous random numbers in the same order.
It has been noted~\cite{Cerutti2008} that this phenomenon leads to unphysical artifacts in simulations.
The simplest approach to properly use PRNGs (and avoid the aforementioned artifacts altogether) consists of reseeding the generator from time to time, obtaining the new seeds from high-quality entropy sources such as those available in modern computer hardware (see~\cite{Dodis2014,Intel2014} for more details). 

\section{Error analysis}\label{sec:error}

\subsection{Stationary distribution error}

In practice, due to time discretization error, 
we cannot generate trajectories 
exactly according to the transition 
kernel $K$. Instead, we can 
generate trajectories according 
to a numerical approximation
$K_\epsilon$. 
We investigate here whether such schemes are consistent, that is, whether powers of $K_\epsilon$ of $K$ 
converge to a distribution $\mu_\epsilon \approx \mu$. 
We emphasize 
that, even though we account for 
time discretization here, we still 
assume infinite sampling, and 
thus for a given $x \in M$, 
$K_\epsilon(x,dy)$ may be a continuous 
distribution. 
See Section~\ref{MFPT_approx}
below for related remarks and 
a discussion of how time discretization 
errors affect the Exact Milestoning 
estimate of the MFPT.

The following theorem, restated from~\cite{Ferre2013},
establishes consistency of iteration 
schemes based on Theorem~\ref{conv} 
when $K_\epsilon$
is sufficiently close to $K$ and $(J_n)$
is geometrically ergodic. 
After the theorem, in Lemma~\ref{geom_ergodic} 
and Theorem~\ref{feller} we give natural 
conditions 
for geometric ergodicity of $(J_n)$.

\begin{theorem}\label{ferre_thm}
Suppose $(J_n)$ is geometrically ergodic:
there exists $\kappa \in (0,1)$
such that
\begin{equation*}
 \sup_{x \in M}||\delta_x K^n - \mu||_{TV}
 =  O(\kappa^n).
\end{equation*}
Let $\{K_\epsilon\}$ be a family of
stochastic kernels with $K_0 = K$, assumed to act
continuously on ${\mathcal B}$, such that
\begin{equation}\label{kernelapprox}
\lim_{\epsilon \to 0} \sup_{|f| \le 1}
\|K_\epsilon f - Kf\|_{\infty} = 0.
\end{equation}
Then for each $\hat \kappa \in (\kappa,1)$, there is $\delta > 0$ such that for
each $\epsilon \in [0,\delta)$,
$K_\epsilon$ has a unique invariant
probability measure $\mu_\epsilon$, and
\begin{align*}
&\sup_{\epsilon < \delta} \sup_{x \in M}
\|\delta_x K_\epsilon^n - \mu_\epsilon \|_{TV} = O({\hat \kappa}^n),\\
&\lim_{\epsilon \to 0} \|\mu_\epsilon -\mu \|_{TV} = 0.
\end{align*}
\end{theorem}

Geometric ergodicity is inconvenient to check
directly. We give two sufficient conditions
for geometric ergodicity of $(J_n)$.
The first condition is a uniform
lower bound on the probability to reach $P$ in $N$
 steps; see Lemma~\ref{geom_ergodic}. We use this to obtain
a strong Feller condition in Theorem~\ref{feller}.
The latter is a very natural condition and is easy to
verify in some cases, for instance when
$(X_t)$ is a nondegenerate diffusion 
and the milestones are sufficiently regular.

\begin{lemma}\label{geom_ergodic}
Suppose that there exists $\lambda \in (0,1)$ and $N \in {\mathbb N}$ such that
for all $x \in M$, ${\mathbb P}^x(J_{N-1} \in P) \ge \lambda > 0$. Then $(J_n)$ is
geometrically ergodic:
\begin{equation*}
 \sup_{x \in M}||\delta_x K^n - \mu||_{TV}
 \le \lambda^{-1}(1-\lambda)^{\lfloor n/N\rfloor}.
\end{equation*}
\end{lemma}
\begin{proof}
Let $\xi_1,\xi_2 \in {\mathcal P}$, consider the
signed measure $\xi = \xi_1-\xi_2$ and compute
\begin{align*}
\|\xi_1 K^N - \xi_2 K^N\|_{TV}
&= \sup_{|f|\le 1}
\left|\int_M \int_M \xi(dy)K^{N}(y,dz)f(z)\right|\\
&= \sup_{|f|\le 1}
\left|\int_M \int_M \xi(dy)\left(K^{N}(y,dz)- \lambda \rho(dz)\right)f(z)\right|\\
&=\sup_{|f|\le 1}
\left|\int_M \xi(dy)\int_M\left(K^{N}(y,dz)- \lambda \rho(dz)\right)f(z)\right|\\
&\le (1-\lambda)\sup_{|f|\le 1}
\left|\int_M \xi(dy)f(y)\right| = (1-\lambda)\|\xi_1-\xi_2\|_{TV}.
\end{align*}
The last line uses the fact that
$K^{N}(y,dz)- \lambda \rho(dz)$
is a positive measure.
This shows that $K^N$ is a contraction mapping on ${\mathcal P}$ with
contraction constant $(1-\lambda)$. 
Observe also that $\|\xi_1 K- \xi_2 K\|_{TV} 
\le \|\xi_1-\xi_2\|_{TV}$. The
result now follows from the contraction
mapping theorem. See for instance Theorem~6.40 
of~\cite{douc2014nonlinear}.
\end{proof}

{Note that the $\lambda$ in 
Lemma}~\eqref{geom_ergodic} {is 
a quantity that can be estimated, 
at least in principle, by 
running trajectories of 
$(X_t)$ starting at $x$ which 
cross $N-1$ milestones before 
reaching $P$. However, this 
is likely impractical for the 
same reason direct estimation 
of the MFPT is impractical -- 
the trajectories would be too 
long. One alternative  
would be to compute the 
probability $P^i(J_{N-1}^{\mathtt{A}} \in P)$ 
for the Markov chain $(J_n^{\mathtt{A}})$ 
on $\{1,\ldots,m\}$
with transition matrix 
${\mathtt{A}}$, and use the 
minimum over $i \in \{1,\ldots,m\}$ 
as a proxy for $\lambda$. Even 
without a practical way to estimate 
$\lambda$, we believe 
the characterization of 
Lemma}~\ref{geom_ergodic} {is useful 
for understanding the 
convergence rate.}

Lemma~\ref{geom_ergodic}
leads to the following
condition for geometric ergodicity of $(J_n)$.

\begin{theorem}\label{feller}
Suppose that $M$ is compact and $(J_n)$ is a strong
Feller chain which is aperiodic
in the sense of~\eqref{nonlattice}.
Then $(J_n)$ is geometrically
ergodic.
\end{theorem}
\begin{proof}
Let $\epsilon \in (0,\mu(P))$.
By Theorem~\ref{conv}, for
each $x \in M$ there is
$N_x \in {\mathbb N}$ such
that ${\mathbb P}^x(J_n \in P) \ge \epsilon$
for all $n \ge N_x$. Because $(J_n)$
is strong Feller, the map
$x \to {\mathbb P}^x(J_n \in P)$
is continuous. By
compactness of $M$, it follows that for
any $\lambda \in (0,\epsilon)$
there is $N \in {\mathbb N}$ such that ${\mathbb P}^y(J_n \in P)
\ge \lambda$ for all $y \in M$
and $n \ge N-1$. Theorem~\ref{geom_ergodic}
now yields the result.
\end{proof}

\subsection{MFPT error}
\label{MFPT_approx}

As discussed above, Equation~\ref{main_eq} can be used to estimate the MFPT
${\mathbb E}^\rho[\tau_P]$
by sampling $\mu$ and local 
MFPTs $\tau_M^x$. The
error in this estimate has two
sources. First, in general we only have
an approximation ${\tilde \mu} \equiv \mu_\epsilon$
of $\mu$. The second source of error is
in the sampling of $\tau_M^x$, due 
to the fact that 
we can only simulate a time discrete
version $({\tilde X}_{n{\delta t}})$
of $(X_t)$. In Theorem~\ref{errors} below 
we give an explicit
formula for the numerical error
of the MFPT in terms of these two sources. 
We first need the following notation. 
Let ${\tilde \tau}_M^x$ be
the minimum of all $n \,\delta t > 0$
such that the line segment between
${\tilde X}_{n{\delta t}}$ and
${\tilde X}_{(n+1){\delta t}}$
intersects $M \setminus M_x$, and define
\begin{equation*}
{\mathbb E}^{\tilde \mu}[{\tilde \tau}_M]
:= \int_M {\tilde \mu}(dx){\mathbb E}^x[{\tilde \tau}_M^x].
\end{equation*}
Theorem~\ref{errors} 
below gives an expression for the error in 
the {\em original} Milestoning as 
well as in Exact Milestoning.

\begin{theorem}\label{errors}
There exists a nonnegative function
$\phi$ such that
\begin{align}\begin{split}\label{error_M}
  \left| {\mathbb E}^\rho[\tau_P] - {\tilde \mu}(P)^{-1}{\mathbb E}^{\tilde \mu}[{\tilde \tau}_M] \right|
  &\le c_1 \left| \mu(P)^{-1}-{\tilde \mu}(P)^{-1} \right| \\
  &\qquad+ {\tilde \mu(P)^{-1}} \left( c_2 \|\mu - {\tilde \mu}\|_{TV} + \phi(\delta t) \right),\end{split}
\end{align}
where
\begin{equation*}
c_1 := {\mathbb E}^\mu[\tau_M],\qquad c_2:=\sup_{x \in M} {\mathbb E}^x[\tau_M^x].
\end{equation*}
\end{theorem}
\begin{proof}
Note that
\begin{align*}
|{\mathbb E}^\rho[\tau_P] -
{\tilde \mu}(P)^{-1} {\mathbb E}^{\tilde \mu}[{\tilde \tau}_M]|&= |\mu(P)^{-1}{\mathbb E}^{\mu}[\tau_M] -
\tilde{\mu}(P)^{-1}{\mathbb E}^{\tilde \mu}[{\tilde \tau}_M]|\\
&\le |\mu(P)^{-1}{\mathbb E}^\mu[\tau_M] -
{\tilde \mu}(P)^{-1}{\mathbb E}^\mu[\tau_M]| \\
&\qquad+ |{\tilde \mu}(P)^{-1}{\mathbb E}^\mu[\tau_M]
- {\tilde \mu}(P)^{-1}{\mathbb E}^{\tilde \mu}[\tau_M]| \\
&\qquad\qquad+ |{\tilde \mu}(P)^{-1}{\mathbb E}^{\tilde \mu}[\tau_M]
- {\tilde \mu}(P)^{-1}{\mathbb E}^{\tilde \mu}[{\tilde \tau}_M]|
\end{align*}
where we have written 
${\mathbb E}^{\tilde \mu}[{ \tau}_M]
:= \int_M {\tilde \mu}(dx){\mathbb E}^x[{\tau}_M^x]$. We may write
\begin{equation*}
\phi(\delta t) = |{\mathbb E}^{\tilde \mu}[\tau_M]
- {\mathbb E}^{\tilde \mu}[{\tilde \tau}_M]|
\end{equation*}
for the term depending only on time stepping error. Note that
\begin{align*}
|{\mathbb E}^\mu[\tau_M] - {\mathbb E}^{\tilde \mu}[\tau_M]|
&= \left|\int_M {\mu}(dx){\mathbb E}^x[\tau_M^x]-
\int_M {\tilde \mu}(dx){\mathbb E}^x[\tau_M^x]\right|\\
&\le \left(\sup_{x \in M} {\mathbb E}^x[\tau_M^x]\right)
\|\mu-{\tilde \mu}\|_{TV}.
\end{align*}
Combining the last three expressions 
yields the result.
\end{proof}

{Recall that in the above we 
have ignored errors from 
finite sampling. We now 
discuss the implications of those 
errors. In the original Milestoning,  
${\tilde \mu}$ is the canonical 
Gibbs distribution on the milestones. 
In that setting, we can typically 
sample independently from ${\tilde \mu}$ 
on the milestones. Thus, the central 
limit theorem implies that 
the} {\em true} {error in the Milestoning 
approximation of ${\mathbb E}^\rho[\tau_P]$ 
is bounded above with high probability 
by the right 
hand side of}~\eqref{error_M} 
{plus a constant times $1/\sqrt{N}$, 
where $N$ is the number of samples. An analogous argument 
applies to Exact Milestoning if 
$\tilde \mu$ is sampled by simple power 
iteration. For our coarse-grained 
version of power iteration in 
Algorithm~\ref{alg:exact-milestoning}, however, 
we obtain samples 
of ${\tilde \mu}$ which are 
not independent, and thus a 
more detailed analysis would 
be required to determine the 
additional error from  
finite sampling.}

We do not analyze the time discretization
error $\phi(\delta t)$ and instead
refer the reader to~\cite{Gobet2010}
and references therein.
Here we simply remark that, if $(X_t)$ is a
diffusion process, then under
certain smoothness assumptions
on the drift and diffusion
coefficients of $(X_t)$ and
on $M$, we have
$\phi(\delta t) = \theta(\sqrt{\delta t})$
when $(X_{n \delta t})$ is the
standard Euler time
discretization with time step $\delta t$.
See~\cite{Gobet2004} for details and proof. 
See also~\cite{Bouchard2013,Higham2013} for numerical schemes that mitigate  
time discretization error in the 
MFPTs.

\section{Illustrative examples}\label{sec:examples}

In this section we discuss two examples of Milestoning to illustrate the method.

We consider the solution, $(X_t)$, of the Brownian dynamics equation,
\begin{equation}
  \label{eq:brownian-dynamics-sde}
  \left\{
    \begin{aligned}
      &\d X_t = -\nabla U(X_t) \, \d t + \sqrt{2 \beta^{-1}} \, \d B_t, \\
      &X_0 \sim \rho
    \end{aligned}
     \right.
\end{equation}
where $U \colon \Omega \to \mathbb{R}$ is a smooth potential energy function, $\beta > 0$ is the inverse temperature, and $(B_t)$ is a standard Brownian motion.

\subsection{M\"uller-Brown potential}
\label{sec:muller-brown}

We begin with a system characterized by the M\"uller-Brown potential~\cite{Muller1979}.
The energy function $U \colon \Omega \subset \mathbb{R}^2 \to \mathbb{R}$ is given by the formula (see also the corresponding energy landscape in Figure~\ref{fig:muller-potential})
\begin{align*}
  U(x_1, x_2) &=
  -200\,{\mathrm{e}^{- \left( {x_1}-1 \right) ^{2}-10\,{{x_2}}^{2}}}
  -100\,{\mathrm{e}^{-{{x_1}}^{2}-10\, \left( {x_2}-\frac{1}{2} \right) ^{2}}} \\
  & -170\,{\mathrm{e}^{-\frac{13}{2}\, \left( {x_1}+\frac{1}{2} \right) ^{2}+11\, \left( {x_1}+\frac{1}{2} \right)  \left( {x_2}-\frac{3}{2} \right) -\frac{13}{2}\, \left( {x_2}-\frac{3}{2} \right) ^{2}}} \\
  &+15\,{\mathrm{e}^{{\frac {7}{10}}\, \left( {x_1}+1 \right) ^{2}+\frac{3}{5}\, \left( {x_1}+1 \right)  \left( {x_2}-1 \right) +{\frac {7}{10}}\, \left( {x_2}-1 \right) ^{2}}}.
\end{align*}
This system is a commonly used benchmark for numerical methods for obtaining reaction rates.
\begin{figure}[ht]
  \centering
  \includegraphics[width=0.7\linewidth]{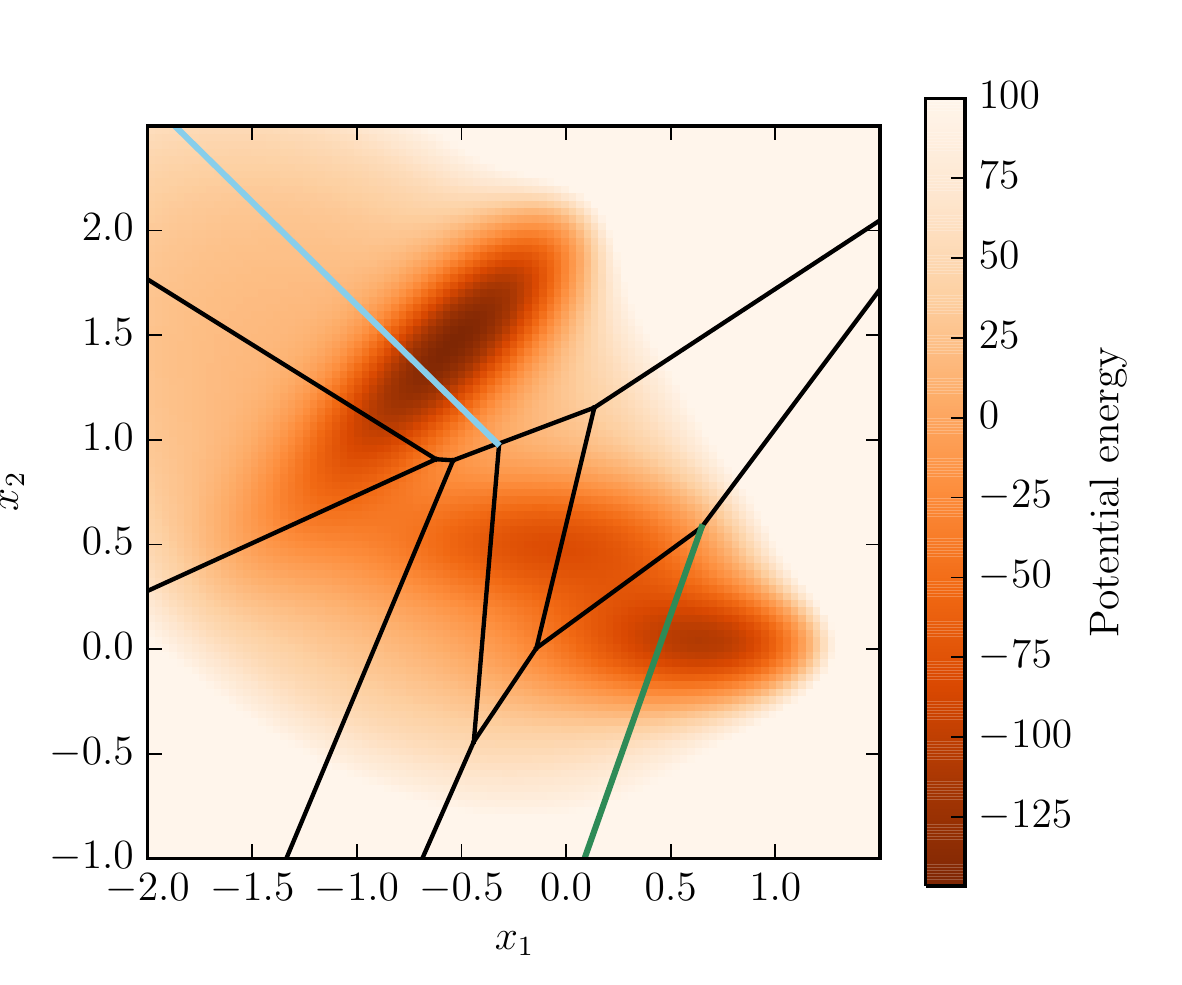}
  \caption{Graph of the M\"uller-Brown potential energy function.  The milestones are shown as the overlaid line segments.}
  \label{fig:muller-potential}
\end{figure}
\begin{figure}[ht]
  \centering
  \includegraphics[width=0.6\linewidth]{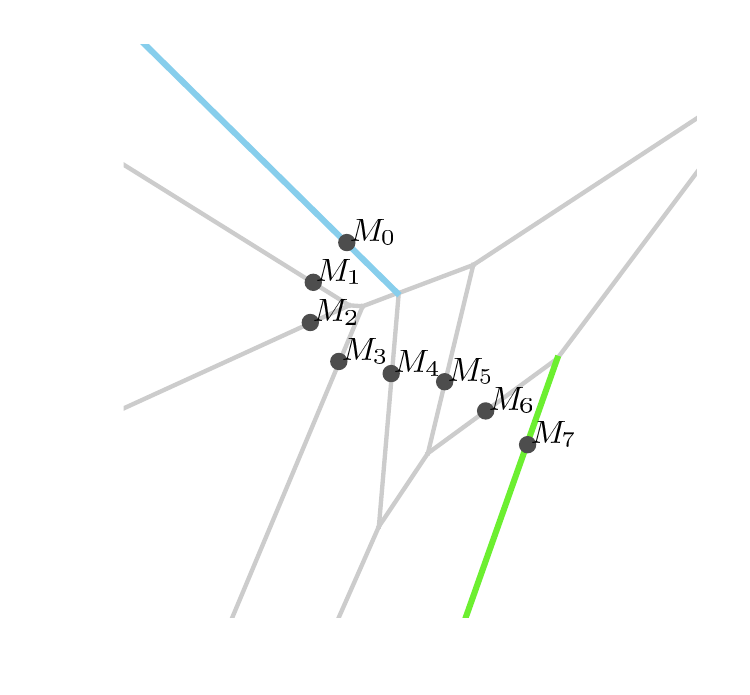}
  \caption{Milestones represented as line segments.  Some of the milestones are labeled and the reactant is colored in green while the product is shown in blue.}
  \label{fig:muller-potential-annotated}
\end{figure}

We chose to partition $\Omega$ using a Voronoi tessellation (displayed in Figures~\ref{fig:muller-potential} and~\ref{fig:muller-potential-annotated}) generated from a set of points gathered by the method of locally updated planes~\cite{Ulitsky1990}.
However, any other set of points could have been chosen (as we shall discuss in the next example).
Figure~\ref{fig:muller-potential} also shows our choice of reactant and product milestones.

For the numerical experiments to be detailed below, we solve the stochastic differential equation in~\eqref{eq:brownian-dynamics-sde} using the Euler-Maruyama scheme~\cite{Milstein2004} with a time step length $\Delta t = 10^{-5}$ at a temperature determined by $\beta^{-1} = 5$.
We use the number of force evaluations as a measure of the computational cost of our methods and we note that the Euler-Maruyama method requires one force evaluation per time step.

We compared two types of experiments that we now describe.
The first experiment consists of running Brownian dynamics trajectories started at the reactant milestone until they reach the product milestone.
As soon as a trajectory reaches the product, we initiate a new trajectory from the reactant milestone and so on.
We refer to these as \emph{long} trajectories.
In the second experiment we run Exact Milestoning starting the first iteration with exactly one phase space point at each of the milestones along the reaction path.
Next, we run ten \emph{short} trajectories per milestone per iteration.
These {short} trajectories start at each milestone and stop whenever they reach any neighboring milestone, as described in Section~\ref{sec:algorithm}.

Despite allowing the long trajectories to go on for approximately $2.5 \times 10^9$ force evaluations, only seven reach the product milestone.
This leads to a poor approximation of the mean first passage time.
By contrast, running the Exact Milestoning method for approximately $2 \times 10^9$ force evaluations, we obtain good estimates of the stationary distribution $\mu$ and the local mean first passage times.

The values of $\mu(M_i)$ are displayed in Figures~\ref{fig:muller-flux} and~\ref{fig:q-muller}.
The empirical distributions corresponding to $\mu$ on some of the milestones are shown in Figure~\ref{fig:empirical-distributions}.
\begin{figure}[!ht]
  \centering
  \subfigure{\includegraphics[width=0.49\linewidth]{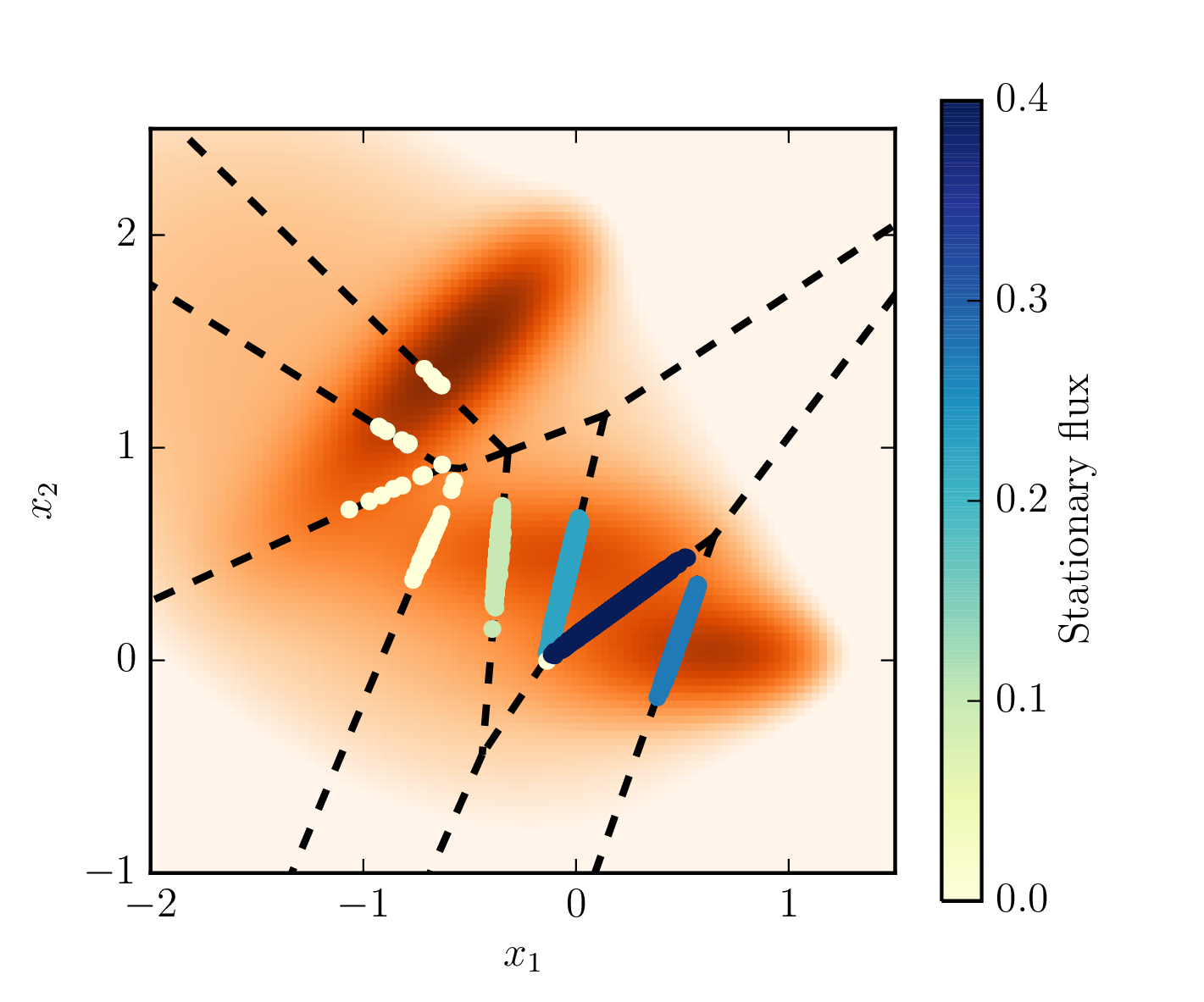}}
  \subfigure{\includegraphics[width=0.49\linewidth]{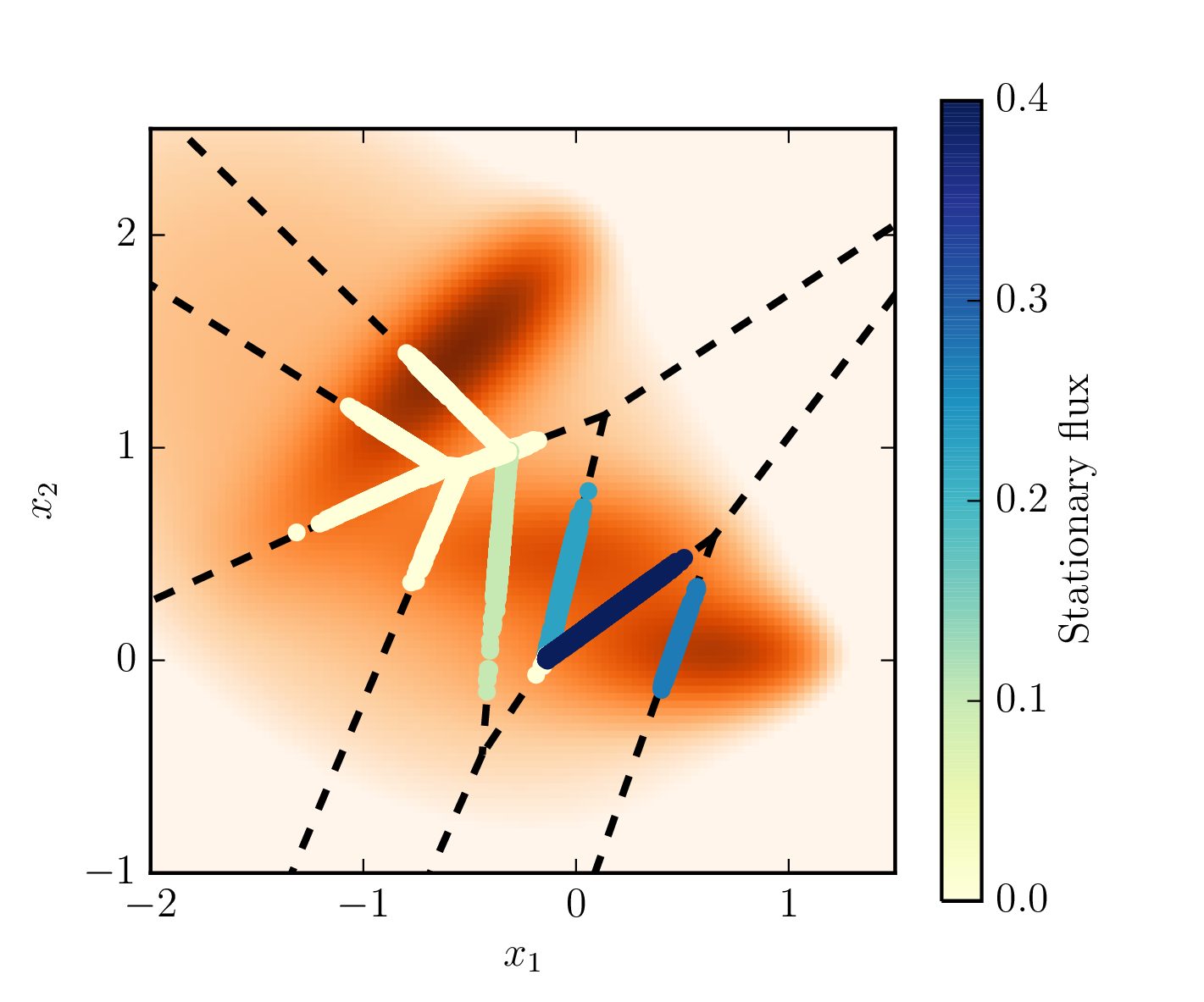}}
  \caption{Phase space points in the empirical distributions of the stationary flux $\mu$ for two types of simulations: long trajectories using straight-forward Brownian dynamics (left) and Exact Milestoning (right).  Despite the fact that the two types of simulations involved comparable amounts of computational effort, we see that the sampling in Exact Milestoning is much more exhaustive.}
  \label{fig:muller-flux}
\end{figure}
\begin{figure}[!ht]
  \centering
  \includegraphics[width=0.75\linewidth]{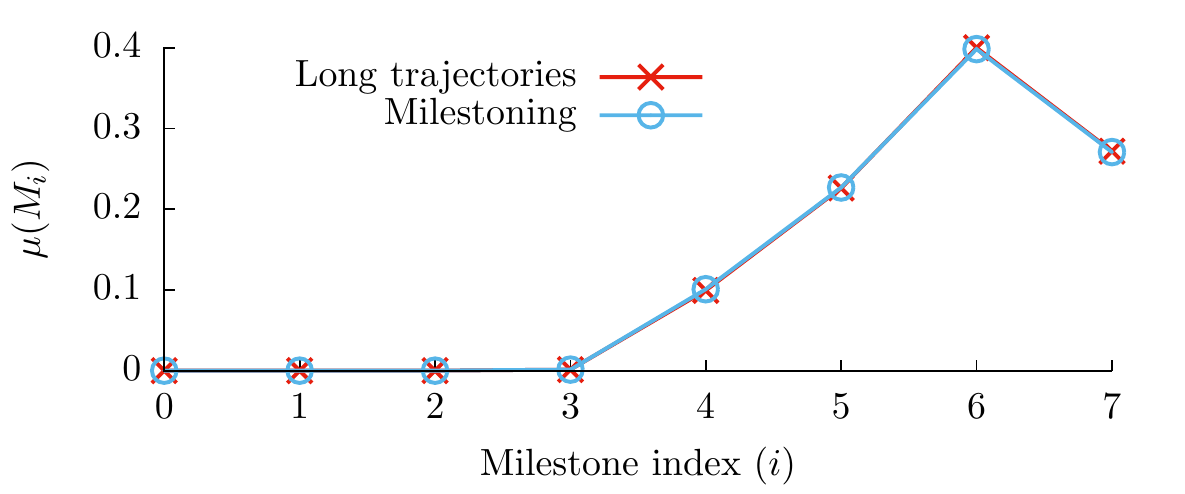}
  \caption{Values of $\mu(M_i)$ at some of the milestones in the M\"{u}ller-Brown potential.
    The values correspond to the long trajectories (in red) and to Exact Milestoning (in blue), as discussed in Section~\ref{sec:muller-brown}.
    Not shown are the milestones other than $M_i$ for $i = 0, \dotsc 7$, where the sampling obtained from the long trajectories is insufficient for comparison.
  }
  \label{fig:q-muller}
\end{figure}
\begin{figure}[!ht]
  \centering
  \includegraphics[width=0.75\linewidth]{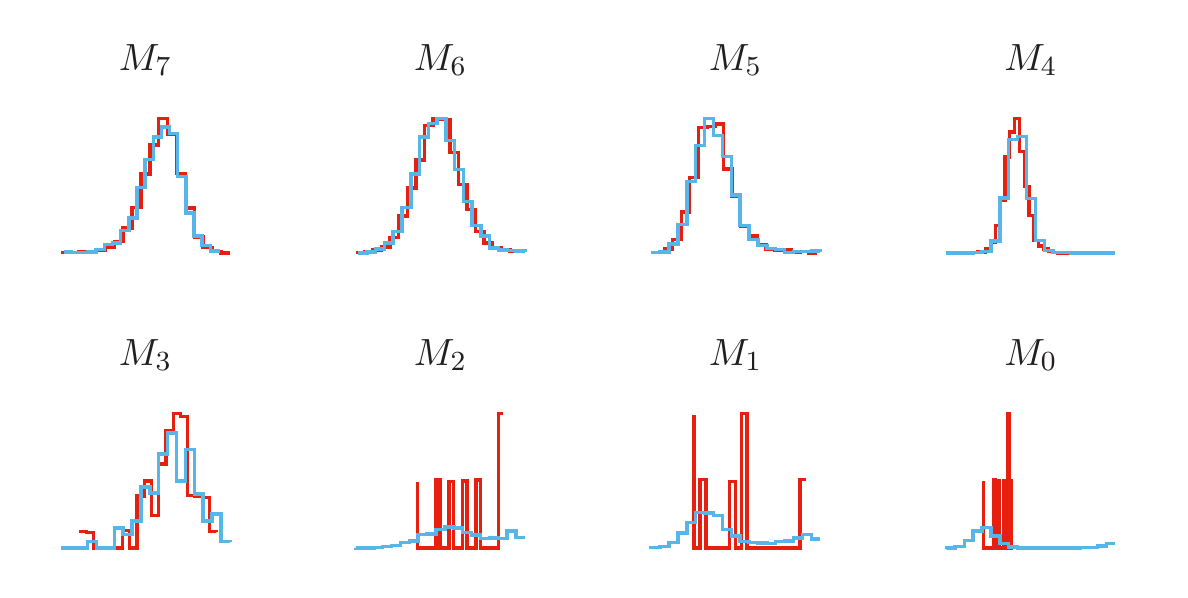}
  \caption{Empirical distributions of the stationary flux obtained by long trajectories (in red) and Exact Milestoning (in blue) corresponding to the system in Section~\ref{sec:muller-brown}. Notice that the sampling of the long trajectories is very sparse at the milestones close to the product.}
  \label{fig:empirical-distributions}
\end{figure}

Figures~\ref{fig:muller-flux} and~\ref{fig:q-muller} illustrate the non-equilibrium nature of Exact Milestoning.
The stationary distribution that we compute differs noticeably from the equilibrium (canonical) distribution.
Recall from Figures~\ref{fig:muller-potential} and~\ref{fig:muller-potential-annotated} that the reactant milestone, $M_7$, is located at the lower-right minimum while the product milestone, $M_0$, is at the global minimum in the upper-left side of the graph.
With this in mind, we see that trajectories initiated at $M_7$ arrive at the intermediate minimum located close to the center of the graph and many of those trajectories return to the lower minimum, crossing $M_7$ again.
This results in high values of $\mu(M_6)$ at the transition state, while the density at $M_0$ (the global minimum) is significantly lower.
Equilibrium considerations, which are inappropriate here, would suggest that most of the stationary density (and the stationary probability) is concentrated at the global minimum and that the weight at the transition state would be small.

\subsection{Rough energy landscape}
\label{sec:rough-landscape}

In this case, we present an example of Milestoning on the torus $\Omega = \mathbb{R}^2 / \, \mathbb{Z}^2$.
For our computations, we consider a uniformly spaced mesh of milestones with fixed product
and reactant sets $P,R$; see Figure~\ref{fig:omega-milestones}.
We model a rough energy landscape by a potential energy function of the form:
\begin{equation}
  \label{eq:energy}
  U(x_1, x_2) = \operatorname{Re} \sum_{k_1 = -N}^{N} \sum_{k_2 = -N}^{N} z_{k_1, k_2} \, \mathrm{e}^{2 \pi \mathrm{i} (k_1 x_1 + \, k_2 x_2)}
\end{equation}
where $\operatorname{Re}$ denotes the real part of a complex number and $N \in \mathbb{N}$ is a constant that tunes the ruggedness of the potential.
Each coefficient $z_{k_1, k_2} = a_{k_1, k_2} + \mathrm{i} \, b_{k_1, k_2} \in \mathbb{C}$ is determined by the random variables $a_{k_1, k_2}$ and $b_{k_1, k_2}$, which are distributed according to
\begin{equation*}
  \begin{cases}
    c, & \text{with probability $\tfrac{1}{2}$}, \\
    0, & \text{with probability $\tfrac{1}{2}$},
  \end{cases}
\end{equation*}
with $c$ itself being a uniform random variable in the interval $(-1, 1)$.
Since $N$ is fixed, a particular realization of the coefficients specified above completely determines the potential energy function $U$.
The graph of the canonical density of a potential energy of the form discussed above is shown in Figure~\ref{fig:rough-energy-landscape}.
\begin{figure}[ht]
  \centering
  \includegraphics[width=\linewidth]{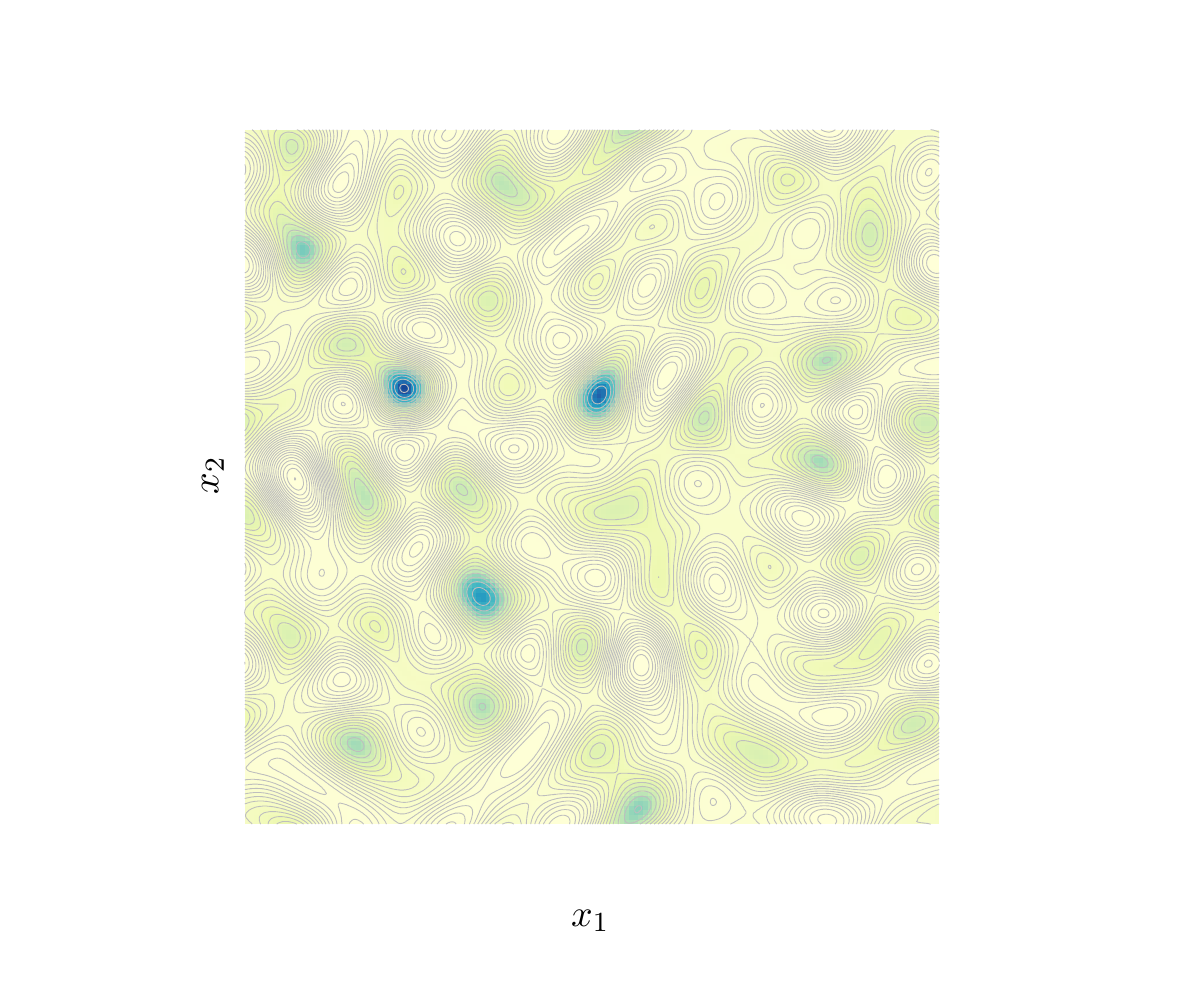}
  \caption{Graph of the density of the canonical distribution corresponding to a rough energy landscape with $N = 7$ at temperature $\beta^{-1} = 1$.}
  \label{fig:rough-energy-landscape}
\end{figure}
Notice that this class of energy functions generalizes the model for rough landscapes introduced by~\cite{Zwanzig1988} and that similar potential energy functions have been used to model Wigner glasses~\cite{Akhanjee2007}.

We carry out Exact Milestoning in this example by solving boundary value problems, as described in~\cite{Bello-Rivas2015b}.
The resulting stationary density obtained after convergence is shown in Figure~\ref{fig:rough-stationary-density}.
\begin{figure}[ht]
  \centering
  \includegraphics[width=0.5\linewidth]{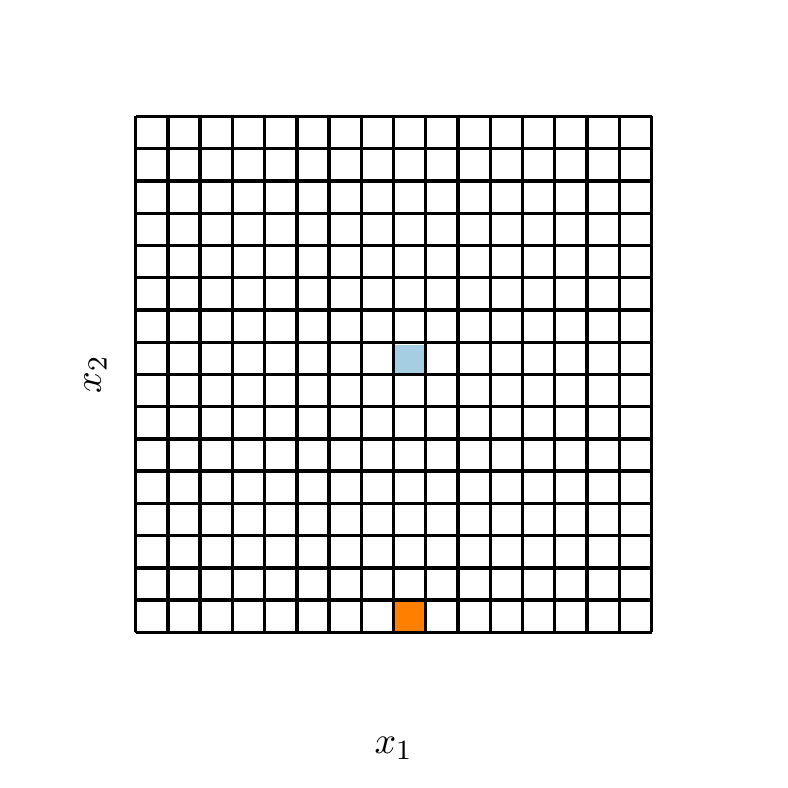}
  \caption{Diagram showing the reactant state (orange square) and the product state (blue square) within the set of all milestones for the example in Section~\ref{sec:rough-landscape}. Each milestone 
  is an edge of one of the small squares 
  in the diagram. (The total number of milestones 
  has been decreased to enhance visibility.)}
  \label{fig:omega-milestones}
\end{figure}
\begin{figure}[ht]
  \centering
  \includegraphics[width=\linewidth]{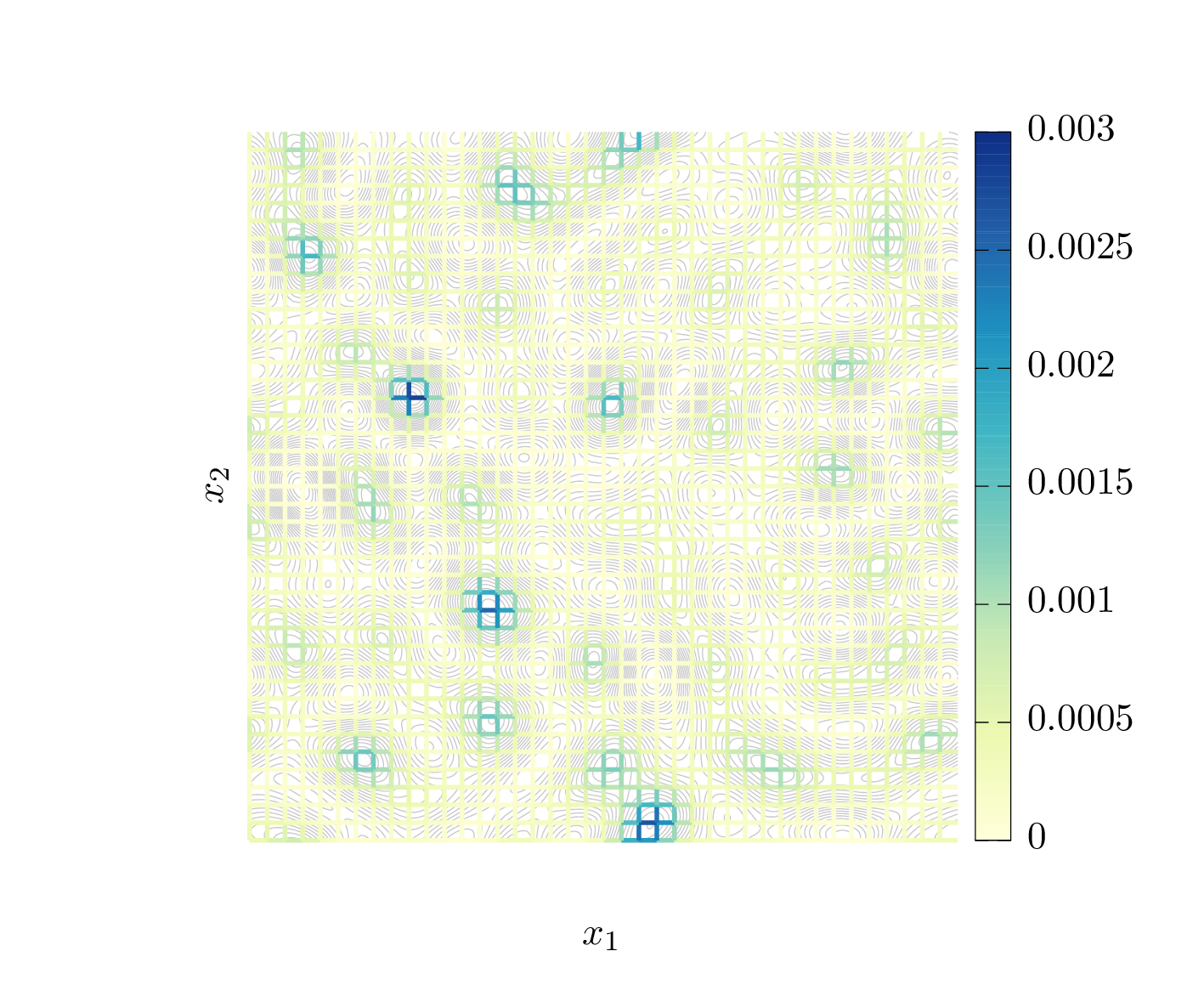}
  \caption{Stationary density $\mu$ on the rough energy landscape of Section~\ref{sec:rough-landscape}.
  The contour lines are the level sets of $U$.
  There are $2\times 40 \times 40$ total milestones (shown as the segments in the overlaid grid).}
  \label{fig:rough-stationary-density}
\end{figure}

It is interesting to note that it has been argued~\cite{Vanden-Eijnden2008} that an optimal choice of milestones would consist of using the level sets (also called \emph{isocommittors}) of the committor function.
These surfaces (see Figure~\ref{fig:isocommittors}) are typically hard to compute in practice, which makes the use of Exact Milestoning more appealing, as its results are independent of how the milestones are set up.
\begin{figure}[ht]
  \centering
  \includegraphics[width=0.6\linewidth]{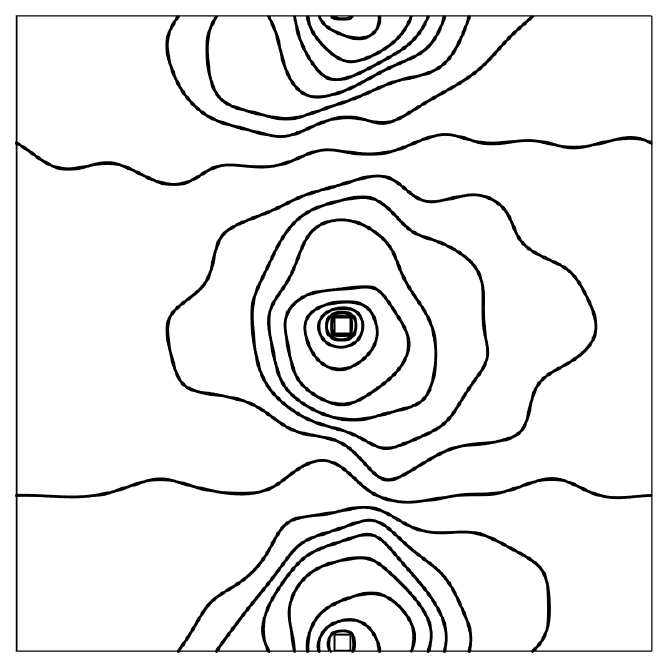}
  \caption{Isocommittor surfaces for the rough energy landscape discussed in Section~\ref{sec:rough-landscape}.}
  \label{fig:isocommittors}
\end{figure}

\section{Conclusions}

The main goal of this manuscript is to present a rigorous mathematical derivation, based on probability theory, of Exact Milestoning. While the theory of Exact Milestoning and accompanying numerical examples were discussed elsewhere~\cite{Bello-Rivas2015,Bello-Rivas2015b}, the mathematical formulation in the earlier paper was not as rigorous as in this manuscript.  Once this formulation is established, it opens the way for further communication between chemical physicists and mathematicians, and it bridges the gap between the communities for further development of an important tool for computer simulation.

Exact Milestoning belongs to a class of enhanced sampling methods for the calculation of kinetics. Most closely related approaches to Milestoning are the Non-Equilibrium Umbrella sampling~\cite{Warmflash2007} and Trajectory Tilting~\cite{Vanden-Eijnden2009}.
The way in which trajectories are sampled is similar in all of these methods; however, the theoretical frameworks are different.
For example, Milestoning allows the calculation of all the moments of the first passage time (FPT) distribution~\cite{Bello-Rivas2015}, and hence better estimates of the FPT can be constructed, a result that was not reported for other methods.

It is not necessary in Milestoning to establish or rely on metastability to estimate the average transition time. From this perspective the method is different from another exact approach ---Transition Path Sampling~\cite{Dellago2002}--- that exploits the short duration of rare trajectories between metastable states. Exact Milestoning makes the sampled trajectories short by sampling trajectory fragments between boundaries of phase space cells or milestones.  The statistics of short trajectories between milestones make it possible to investigate wide ranges of types of energy landscapes, which may be corrugated or not, as illustrated in the two examples in this manuscript. We have shown in the present manuscript that Exact Milestoning is both accurate and highly efficient.

It is important to emphasize that the choice of the milestones in Exact Milestoning is arbitrary from a formal viewpoint. Efficiency considerations suggest that it is beneficial to select them following two criteria: (i) the milestones should be sufficiently close in the kinetic sense to make the trajectories short, and (ii) milestones should be chosen to make the number of iterations as small as possible. For example, the number of iterations can be small if the system is close to equilibrium and the milestones are expressed in a space of slow variables. Then an initial choice of the canonical distribution is quite accurate.

We comment that the method of Milestoning that was broadly used in the past ({e.g.,} \cite{Kirmizialtin2011}) is approximate and assumes local equilibrium within the milestones. While corrections and further refinements were proposed~\cite{Majek2010,Hawk2011,Hawk2013}, these approximations cannot be made exact and are similar in spirit to the local equilibrium and lag time approximations of Markov State Models~\cite{Chodera2007}. Nevertheless, these approximations can be accurate with a proper choice of coarse variables. These types of approximations are very useful as the system grows in complexity and size and exact calculations become prohibitively expensive. Milestoning made it possible to investigate kinetics of enzymes~\cite{Kirmizialtin2015} and transport through membranes~\cite{Cardenas2015} in agreement with experimental observations. These are systems of tens to hundreds of thousands of particles and time scales of milliseconds. It will be of considerable interest to re-evaluate these approximations for large systems with the method of Exact Milestoning. As the efficiency of Exact Milestoning increases with faster hardware, we are breaking scale barriers that were not accessible before to atomically detailed simulations.

\section{Acknowledgments}
\label{sec:acknowledgments}

This research was supported in part by a grant from the NIH GM59796
and from the Welch Foundation Grant No. F-1783. 
D. Aristoff gratefully acknowledges support from 
the National Science Foundation via the award
NSF-DMS-1522398.


\begin{thebibliography}{99}

\bibitem{Akhanjee2007}
{\sc S.~Akhanjee and J.~Rudnick}, {\em {Disorder induced transition into a
  one-dimensional Wigner glass}}, Physical Review Letters, 99 (2007),
  pp.~236403, 0706.4462.

\bibitem{Allen2006} {\sc R.J. Allen, D. Frenkel, and P.R. ten Wolde}, {em Forward flux sampling-type schemes for simulating rare events: Efficiency analysis}, The Journal of Chemical Physics, 124(19) (2006), pp.~463102. 

\bibitem{Allen1989}
{\sc P.~Allen and D.~J. Tildesley}, {\em {Computer simulation of liquids}},
  Oxford Science Publications, Clarendon Press, 1989.
  
  \bibitem{Alsmeyer1994} {\sc G. Alsmeyer}, 
  {\em On the Markov renewal theorem}, 
  Stochastic Processes and their Applications, 
  50 (1994), pp.~37--56.

\bibitem{Bello-Rivas2015}
{\sc J.~M. Bello-Rivas and R.~Elber}, {\em {Exact milestoning}}, The Journal of
  Chemical Physics, 142 (2015), pp.~094102.

\bibitem{Bello-Rivas2015b}
{\sc J.~M. Bello-Rivas and R.~Elber}, {\em {Simulations of
  thermodynamics and kinetics on rough energy landscapes with milestoning}},
  Journal of Computational Chemistry, (2015).
  
  \bibitem{Bolhuis2002} {\sc P.G. Bolhuis et al.}, {\em Transition path sampling: Throwing ropes over roughmountain passes, in the dark.} Annual Review of Physical Chemistry, 53 (2002) pp.~291-318. 

\bibitem{Bouchard2013}
{\sc B.~Bouchard, S.~Geiss, and E.~Gobet}, {\em {First time to exit of a
  continuous It\^{o} process: general moment estimates and $L_1$-convergence
  rate for discrete time approximations}},  (2013), pp.~1--31, 1307.4247.

\bibitem{Cardenas2012}
{\sc A.E. Cardenas et al.}, {\em Unassisted Transport of N-Acetyl-L-tryptophanamide through Membrane: Experiment and Simulation of Kinetics}, Journal of Physical Chemistry B, 116(9) (2012), pp.~2739--2750.
  
  \bibitem{Cardenas2012_2}
 {\sc A.E. Cardenas and R. Elber}, {\em Enhancing the capacity of molecular dynamics simulations with trajecory fragments, in Innovation in Biomolecular Modeling and Simulation}, T. Schlick, Editor. 2012, Royal Society of Chemistry: London.
 
 \bibitem{Cardenas2013} {\sc A.E. Cardenas and R. Elber}, {\em Computational study of peptide permeation through membrane: searching for hidden slow variables}, Molecular Physics, 111 (2013), pp.~356--3578.

\bibitem{Cardenas2015}
{\sc A.~E. Cardenas, R.~Shrestha, L.~J. Webb, and R.~Elber}, {\em {Membrane
  Permeation of a Peptide: It Is Better to be Positive}}, The Journal of
  Physical Chemistry B, 119 (2015), pp.~6412--6420.

\bibitem{Cerutti2008}
{\sc D.~S. Cerutti, R.~Duke, P.~L. Freddolino, H.~Fan, and T.~P. Lybrand}, {\em
  {A vulnerability in popular Molecular Dynamics packages concerning Langevin
  and Andersen dynamics}}, Journal of Chemical Theory and Computation, 4
  (2008), pp.~1669--1680.

\bibitem{Chodera2007}
{\sc J.~D. Chodera, N.~Singhal, V.~S. Pande, K.~A. Dill, and W.~C. Swope}, {\em
  {Automatic discovery of metastable states for the construction of Markov
  models of macromolecular conformational dynamics}}, The Journal of Chemical
  Physics, 126 (2007), pp.~155101.

\bibitem{Dellago1998} {\sc C. Dellago et al.}, 
{\em Transition path sampling and the calculation of rate constants}, The Journal of Chemical Physics, 
108(5) (1998), pp.~1964--1977.

\bibitem{Dellago2002}
{\sc C.~Dellago, P.~G. Bolhuis, and P.~L. Geissler}, {\em {Transition Path
  Sampling}}, vol.~123, 2002.

\bibitem{Dodis2014}
{\sc Y.~Dodis, A.~Shamir, N.~Stephens-Davidowitz, and D.~Wichs}, {\em {How to
  eat your entropy and have it too --- Optimal recovery strategies for
  compromised RNGs}}, in Advances in Cryptology – CRYPTO 2014 SE - 3,
  J.~Garay and R.~Gennaro, eds., vol.~8617 of Lecture Notes in Computer
  Science, Springer Berlin Heidelberg, 2014, pp.~37--54.

\bibitem{douc2014nonlinear}
{\sc R.~Douc, E.~Moulines, and D.~Stoffer}, {\em {Nonlinear time series:
  theory, methods and applications with R examples}}, Chapman \& Hall/CRC Texts
  in Statistical Science, Taylor \& Francis, 2014.

\bibitem{Elber2007} {\sc Elber, R.}, 
{\em A milestoning study of the kinetics of an allosteric transition: Atomically detailed simulations of deoxy Scapharca hemoglobin}, Biophysical Journal, 92(9) (2007), pp. L85--L87. 

\bibitem{Elber2009} {\sc R. Elber, K. Kuczera, and G.S. Jas}, {\em The kinetics of helix unfolding: Molecular dynamics simulations with Milestoning}, Journal of Physical Chemistry A, 113 (2009), pp. 7461-7473. 


\bibitem{Elber2010} {\sc R. Elber and A. West}, {\em Atomically detailed simulation of the recovery stroke in myosin by Milestoning}, Proceedings of the National Academy of Sciences USA 107, (2010), pp.~5001-5005.

\bibitem{Faradjian2004}
{\sc A.~K. Faradjian and R.~Elber}, {\em {Computing time scales from reaction
  coordinates by Milestoning}}, Journal of Chemical Physics, 120 (2004),
  pp.~10880--10889.

\bibitem{Ferre2013}
{\sc D.~Ferr\'{e}, L.~Herv\'{e}, and J.~Ledoux}, {\em {Regular perturbation of
  V-geometrically ergodic Markov chains}}, Journal of Applied Probability, 50
  (2013), pp.~184--194.

\bibitem{Frenkel2002}
{\sc D.~Frenkel and B.~Smit}, {\em {Understanding molecular simulation: from
  algorithms to applications}}, 2002.
  
 \bibitem{Glowacki2011} {\sc D.R. Glowacki, E. Paci, and D.V. Shalashilin}, {\em Boxed Molecular Dynamics: Decorrelation Time Scales and the Kinetic Master Equation}, Journal of Chemical Theory and Computation, 7(5) (2011), pp.~1244-1252.

\bibitem{Gobet2004}
{\sc E.~Gobet and S.~Menozzi}, {\em {Exact approximation rate of killed
  hypoelliptic diffusions using the discrete Euler scheme}}, Stochastic
  Processes and their Applications, 112 (2004), pp.~201--223.

\bibitem{Gobet2010}
\leavevmode\vrule height 2pt depth -1.6pt width 23pt, {\em {Stopped diffusion
  processes: Boundary corrections and overshoot}}, Stochastic Processes and
  their Applications, 120 (2010), pp.~130--162, 0706.4042.

\bibitem{Golub2013}
{\sc G.~H. Golub and C.~F. {Van Loan}}, {\em {Matrix computations}}, Johns
  Hopkins Studies in the Mathematical Sciences, Johns Hopkins University Press,
  Baltimore, MD, fourth~ed., 2013.

\bibitem{Hawk2013}
{\sc A.~T. Hawk}, {\em {Milestoning with coarse memory}}, The Journal of
  chemical physics, 138 (2013), pp.~154105.

\bibitem{Hawk2011}
{\sc A.~T. Hawk and D.~E. Makarov}, {\em {Milestoning with transition memory}},
  The Journal of Chemical Physics, 135 (2011), pp.~224109.

\bibitem{Higham2013}
{\sc D.~J. Higham, X.~Mao, M.~Roj, Q.~Song, and G.~Yin}, {\em {Mean exit times
  and the Multilevel Monte Carlo method}}, SIAM/ASA Journal on Uncertainty
  Quantification, 1 (2013), pp.~2--18.
  
  \bibitem{Huber1996} {\sc G.A. Huber and S. Kim}, {\em Weighted-ensemble Brownian dynamics simulations for protein association reactions}, Biophysical Journal, 70(1) (1996), pp.~97--110. 

\bibitem{Intel2014}
{\sc {Intel Corporation}}, {\em {Intel\textsuperscript{\textregistered} Digital
  Random Number Generator (DRNG) Software Implementation Guide}}, 2014.

\bibitem{Jas2012}
{\sc G.S. Jas et al.}, {\em Experiments and Comprehensive Simulations of the Formation of a Helical Turn}, 
Journal of Physical Chemistry B, 116(23) (2012), pp.~6598-6610. 


\bibitem{Kirmizialtin2011}
{\sc S.~Kirmizialtin and R.~Elber}, {\em {Revisiting and computing reaction
  coordinates with Directional Milestoning}}, The journal of physical
  chemistry. A, 115 (2011), pp.~6137--48.

\bibitem{Kirmizialtin2012} 
{\sc S. Kirmizialtin et al.}, 
{\em How Conformational Dynamics of DNA Polymerase Select Correct Substrates: Experiments and Simulations}, Structure, 20(4) (2012), pp.~618-627. 

\bibitem{Kirmizialtin2015}
{\sc S.~Kirmizialtin, K.~A. Johnson, and R.~Elber}, {\em {Enzyme Selectivity of
  HIV Reverse Transcriptase: Conformations, Ligands, and Free Energy
  Partition}}, The Journal of Physical Chemistry B, 119 (2015),
  pp.~11513--11526.
  
\bibitem{Kreuzer2012} 
{\sc  S.M. Kreuzer, R. Elber, and T.J. Moon}, 
{\em Early Events in Helix Unfolding under External Forces: A Milestoning Analysis}, The Journal of Physical Chemistry B, 116(29) (2012), pp. 8662--8691.

\bibitem{Lehoucq1998}
{\sc R.~B. Lehoucq, D.~C. Sorensen, and C.~Yang}, {\em A{RPACK} users' guide},
  vol.~6 of Software, Environments, and Tools, Society for Industrial and
  Applied Mathematics (SIAM), Philadelphia, PA, 1998.

\bibitem{Majek2010}
{\sc P.~M{\'{a}}jek and R.~Elber}, {\em {Milestoning without a reaction
  coordinate}}, Journal of chemical theory and computation, 6 (2010),
  pp.~1805--1817.

\bibitem{Metzner2009}
{\sc P.~Metzner, C.~Sch\"{u}tte, and E.~Vanden-Eijnden}, {\em {Transition Path
  Theory for Markov jump processes}}, 2009.

\bibitem{Milstein2004}
{\sc G.~N. Milstein and M.~V. Tretyakov}, {\em {Stochastic numerics for
  Mathematical Physics}}, Scientific Computation, Springer Berlin Heidelberg,
  Berlin, Heidelberg, 2004.

\bibitem{Moroni2004}
{\sc D.~Moroni, P.~G. Bolhuis, and T.~S. {Van Erp}}, {\em {Rate constants for
  diffusive processes by partial path sampling}}, Journal of Chemical Physics,
  120 (2004), pp.~4055--4065, 0310466.

\bibitem{Muller1979}
{\sc K.~M{\"{u}}ller and L.~D. Brown}, {\em {Location of saddle points and
  minimum energy paths by a constrained simplex optimization procedure}},
  Theoretica Chimica Acta, 53 (1979), pp.~75--93.

\bibitem{Ruymgaart2011}
{\sc A.~P. Ruymgaart, A.~E. Cardenas, and R.~Elber}, {\em {MOIL-opt:
  Energy-conserving molecular dynamics on a GPU/CPU system}}, Journal of
  chemical theory and computation, 7 (2011), pp.~3072--3082.

\bibitem{Sarich2010}
{\sc M.~Sarich, F.~No\'{e}, and C.~Sch\"{u}tte}, {\em {On the approximation
  quality of Markov State Models}}, 2010.

\bibitem{Schlick2010}
{\sc T.~Schlick}, {\em {Molecular modeling and simulation: an interdisciplinary
  guide}}, 2010.

\bibitem{Shaw2008}
{\sc D.~E. Shaw, J.~C. Chao, M.~P. Eastwood, J.~Gagliardo, J.~P. Grossman,
  C.~R. Ho, D.~J. Lerardi, I.~Kolossv\'{a}ry, J.~L. Klepeis, T.~Layman,
  C.~McLeavey, M.~M. Deneroff, M.~A. Moraes, R.~Mueller, E.~C. Priest, Y.~Shan,
  J.~Spengler, M.~Theobald, B.~Towles, S.~C. Wang, R.~O. Dror, J.~S. Kuskin,
  R.~H. Larson, J.~K. Salmon, C.~Young, B.~Batson, and K.~J. Bowers}, {\em
  {Anton, a special-purpose machine for molecular dynamics simulation}}, 2008.

\bibitem{Stone2007}
{\sc J.~E. Stone, J.~C. Phillips, P.~L. Freddolino, D.~J. Hardy, L.~G. Trabuco,
  and K.~Schulten}, {\em {Accelerating molecular modeling applications with
  graphics processors}}, Journal of Computational Chemistry, 28 (2007),
  pp.~2618--2640.

\bibitem{Swenson2014}
{\sc D.~W.~H. Swenson and P.~G. Bolhuis}, {\em {A replica exchange transition
  interface sampling method with multiple interface sets for investigating
  networks of rare events}}, The Journal of Chemical Physics, 141 (2014).

\bibitem{Ulitsky1990}
{\sc A.~Ulitsky and R.~Elber}, {\em {A new technique to calculate steepest
  descent paths in flexible polyatomic systems}}, The Journal of Chemical
  Physics, 92 (1990), pp.~1510.
  
\bibitem{van-Erp} {\sc an Erp, T.S., D. Moroni, and P.G. Bolhuis}, {\em A novel path sampling method for the calculation of rate constants}, The Journal of Chemical Physics, 118(17) (2003), pp.~7762--7774.

\bibitem{Vanden-Eijnden2009}
{\sc E.~Vanden-Eijnden and M.~Venturoli}, {\em {Exact rate calculations by
  trajectory parallelization and tilting}}, The Journal of chemical physics,
  131 (2009), pp.~1--7, 0904.3763.

\bibitem{Vanden-Eijnden2008}
{\sc E.~Vanden-Eijnden, M.~Venturoli, G.~Ciccotti, and R.~Elber}, {\em {On the
  assumptions underlying milestoning}}, The Journal of chemical physics, 129
  (2008), pp.~174102.

\bibitem{Warmflash2007}
{\sc A.~Warmflash, P.~Bhimalapuram, and A.~R. Dinner}, {\em {Umbrella sampling
  for nonequilibrium processes}}, The Journal of chemical physics, 127 (2007),
  pp.~154112.
  
  \bibitem{West2007} 
 {\sc A.M.A West, R. Elber, and D. Shalloway}, {\em Extending molecular dynamics time scales with milestoning: Example of complex kinetics in a solvated peptide}, 
Journal of Chemical Physics, 126(14) (2007). 

\bibitem{Zhang2010} 
{\sc B.W. Zhang, D. Jasnow, and D.M. Zuckerman}, 
{\em The ``weighted ensemble'' path sampling method is statistically exact for a broad class of stochastic processes and binning procedures}, Journal of Chemical Physics, 132(5) (2010). 


\bibitem{Zwanzig1988}
{\sc R.~Zwanzig}, {\em {Diffusion in a rough potential}}, Proceedings of the
  National Academy of Sciences of the United States of America, 85 (1988),
  pp.~2029--2030.

\end{thebibliography}
\end{document}